\documentclass[sigconf, nonacm]{acmart}
\usepackage{graphicx} % Required for inserting images
\usepackage{amsthm}
\usepackage{enumitem}
\usepackage{colortbl}
\usepackage{xcolor}
\usepackage{listings}
\usepackage{multirow}
\usepackage{makecell}
\usepackage{subcaption}
\usepackage{color}
\usepackage{xcolor}
\usepackage{xspace}
\usepackage{algorithm} 
\usepackage{algorithmic} 
\usepackage{wrapfig}

\newtheorem{theorem}{Theorem}
\newtheorem{lemma}{Lemma}

\usepackage{nicematrix}

\newcommand{\mypara}[1]{\noindent\textbf{#1}}
\newcommand{\mysubpara}[1]{\noindent\emph{#1}}

\definecolor{mygray}{rgb}{0.92,0.92,0.92}
\newcommand{\codeinline}[1]{\colorbox{mygray}{\lstinline|#1|}}

\title{\textsf{VertiMRF}: Differentially Private Vertical Federated Data Synthesis}

\author{Fangyuan Zhao}
\affiliation{%
  \institution{Xi'an Jiaotong University}
}
\email{zfy1454236335@stu.xjtu.edu.cn}

\author{Zitao Li}
\affiliation{%
  \institution{Alibaba Group}
}
\email{zitao.l@alibaba-inc.com}

\author{Xuebin Ren}
\affiliation{%
  \institution{Xi'an Jiaotong University}
}
\email{xuebinren@mail.xjtu.edu.cn}

\author{Bolin Ding}
\affiliation{%
  \institution{Alibaba Group}
}
\email{bolin.ding@alibaba-inc.com}

\author{Shusen Yang}
\affiliation{%
  \institution{Xi'an Jiaotong University}
}
\email{shusenyang@mail.xjtu.edu.cn}

\author{Yaliang Li}
\affiliation{%
  \institution{Alibaba Group}
}
\email{yaliang.li@alibaba-inc.com}

\begin{abstract}
Data synthesis is a promising solution to share data for various downstream analytic tasks without exposing raw data. However, without a theoretical privacy guarantee, a synthetic dataset would still leak some sensitive information. Differential privacy is thus widely adopted to safeguard data synthesis by strictly limiting the released information.
This technique is advantageous yet presents significant challenges in the vertical federated setting, where data attributes are distributed among different data parties. 
The main challenge lies in maintaining privacy while efficiently and precisely reconstructing the correlation among cross-party attributes.
In this paper, we propose a novel algorithm called \textsf{VertiMRF}, designed explicitly for generating synthetic data in the vertical setting and providing differential privacy protection for all information shared from data parties. 
We introduce techniques based on the Flajolet-Martin sketch (or frequency oracle) for encoding local data satisfying differential privacy and estimating cross-party marginals.
We provide theoretical privacy and utility proof for encoding in this multi-attribute data.
Collecting the locally generated private Markov Random Field (MRF) and the sketches, a central server can reconstruct a global MRF, maintaining the most useful information.
Additionally, we introduce two techniques tailored for datasets with large attribute domain sizes, namely dimension reduction and consistency enforcement.
These two techniques allow flexible and inconsistent binning strategies of local private MRF and the data sketching module, which can preserve information to the greatest extent.
We conduct extensive experiments on four real-world datasets to evaluate the effectiveness of \textsf{VertiMRF}. 
End-to-end comparisons demonstrate the superiority of \textsf{VertiMRF}, and ablation studies validate the effectiveness of each component.
    
\end{abstract}

\begin{document}
\sloppy

\maketitle

\section{Introduction}

%zitao: privacy,  vertical, 

% With the increasing stringency of data privacy regulations such as European General
% Data Protection Regulation (GDPR)~\cite{voigt2017eu} and the California Consumer Privacy Act~\cite{pardau2018california}, data privacy has become a significant concern for various data analysis tasks. In this context, private data synthesis has emerged as a promising technique. Its aim is to generate and release synthetic data that preserves the statistical characteristics of the original data, allowing for diverse data analysis tasks to be conducted without compromising data privacy. 
% The popularity of private data synthesis arises from the fact that the released synthetic data can support unrestricted queries and computations without requiring additional design and interaction tailored to specific data analysis tasks. However, ensuring strict privacy guarantees while generating synthetic data of high quality remains a challenge. 

With the increasing stringency of data privacy regulations such as the European General Data Protection Regulation (GDPR)~\cite{voigt2017eu} and the California Consumer Privacy Act~\cite{pardau2018california}, data privacy has become a significant concern for various data analysis tasks. 
Following this trend, data synthesis has emerged as a promising technique. 
For the tabular data domain, the synthesis algorithms aim to generate and release synthetic data that preserves the statistical characteristics of the original data, allowing for diverse data analysis tasks to be conducted without access to the original real data from individuals. 

Coupled with \textit{differential privacy} (DP)~\cite{dwork2006calibrating,zhang2018differentially,zhao2020latent,ren2024belt} techniques, the synthetic data can provide theoretical privacy guarantees for arbitrary individual records in the original datasets.
Compared with other DP algorithms for specific analytic tasks, DP data synthesis can support an unlimited number of unrestricted downstream tasks without additional privacy loss other than the one occurring during data synthesis~\cite{hu2023sok}.
The main challenge emerges when ensuring DP while generating synthetic data of high quality. A growing body of academic research~\cite{abay2019privacy,zhang2018differentially,frigerio2019differentially,chen2020gs,jordon2018pate, mckenna2019graphical,mckenna2021winning,zhang2017privbayes,cai2021data,zhang2021privsyn,ren2022ldp,wang2019locally} has focused on improving the trade-off between privacy and utility of DP synthetic data and already obtained promising results. However, these studies primarily focus on the centralized setting, assuming that the raw data has already been collected by a trusted curator.

To realize the value of data at the furthest level, multiple data parties may want to cooperate on some tasks for more comprehensive and accurate information.
If such cooperation is achieved without sharing data directly, the setting is generally called federated learning (FL)~\cite{kairouz2021advances,xie2023federatedscope,li2024performance}.
A relatively well-studied scenario in FL is that data parties have data with the same set of attributes but from different groups of individuals.
This scenario is called horizontal federated learning (HFL) because the local dataset can be obtained by splitting a virtual global dataset by individuals~\cite{mcmahan2018dp-rnn, mcmahan2017dpfedavg}. Under such a setting, several studies have focused on DP data synthesis under the horizontally-distributed~\cite{su2016differentially} and local DP settings~\cite{ren2018lopub,zhang2018calm}.
Nevertheless, another attractive but challenging case is when data parties have data from the same set of individuals but on different attributes~\cite{liu2020federated-forest,hu2019fdml, li2022vldb, li2023differentially}.
Symmetrically, this setting is called vertical federated learning (VFL) as local datasets can be derived by dividing a global dataset by attributes.
VFL techniques attract the attention of many medical or fintech companies~\cite{webank} because their model accuracy can be boosted by more comprehensive information brought by VFL.
% For example, Webank~\cite{webank} combines users' basic profile data from internet companies with their behavioral data from car rental companies to enhance risk control measures.
In this paper, we focus on data synthesis in the VFL setting as it has great potential in various aspects. 

% Recently, with the development of data processing technologies, the limited dimensions of data in a data party are insufficient to meet the demands of high-quality data analysis. Often, it is necessary to integrate diverse feature dimensions from multiple parties to enable comprehensive data analysis and high-quality AI model training, i.e., federated learning (FL). For example, Webank combines users' basic profile data from internet companies with their behavioral data from car rental companies to enhance risk control measures. These scenarios, where the same set of user attributes is distributed across different parties, are commonly known as the \textit{vertical setting}. Generating synthetic data over the distributed attributes under \textit{vertical setting} has great potential in various aspects. 
% \begin{itemize}[itemindent=2ex, listparindent=0ex, leftmargin=0ex, labelwidth=0px]
\mysubpara{1) It facilitates the cross-party data analysis.} Simply combining the synthetic data generated independently by the multiple parties loses the statistical property of cross-party attributes. However, when a VFL data synthesis algorithm that accurately captures the cross-party correlation is available, any downstream correlation analysis can be done efficiently once the synthetic data is ready. 
% \mysubpara{2) It enables hyper-parameter selection.} 
\mysubpara{2) It enable validating or tuning general VFL algorithms under controllable privacy risk.}
For example, VFL tasks often involve substantial costs for hyperparameter tuning among multi-parties, due to the strict limitations of cross-party data access. Releasing a synthetic dataset that preserves the statistical characteristics of the original data can help select optimal hyper-parameters before model training.

Despite the great potential, there are following \textit{challenges} that hinder the practical applications.

\mysubpara{C1: Information loss when estimating cross-party attribute correlations.}
Unlike algorithms in the central setting that can access all data attributes, VFL synthesis algorithms that can faithfully generate data in global-view must have components to estimate the correlation of the cross-party attributes, either explicitly or implicitly.
However, such estimation must suffer information loss because of either the distillation of raw data or added randomness for privacy.

% \mysubpara{C2: Information loss when estimating cross-party attribute correlations.}
% For intra-party attributes, correlations can be accurately estimated by applying existing centralized DP methods. While for cross-party attributes, there is an additional constraint, i.e., no access to other parties' data, which has each party to share noisy information to estimate cross-party correlations. That would incur more information losses than estimating intra-party correlations.
\mysubpara{C2: Composing and trade-off the intra-party and cross-party information.}
It is known that statistics estimated in the central DP setting can have higher accuracy than the same ones obtained in the distributed DP settings. 
Although the intuitive idea following this is to utilize as much information as possible that does not rely on cross-party cooperation, how to effectively and efficiently combine and balance this information with estimated cross-party correlation information remains to be explored.

\mysubpara{C3: Curse of dimensionality.} In VFL settings, a record may contain multiple attributes that distributed among multiple parties, each attribute with large domain size. In this case, there are multiple cross-party attribute combinations to estimate, which would introduce overwhelming noises and huge communication costs. 

% However, this is challenging in the vertical setting because each party only has access to his/her local attributes, and the server can only access the information shared by data parties.
% 1) How to estimate cross-party attribute correlations with comparable utility and provable privacy guarantee? 2) How to merge and tradeoff the intra-party and cross-party attribute correlations estimated with different methods while preserving as more useful information as possible? 3) How to tackle the curse of dimensionality, i.e., the overwhelming noises, when domain sizes of attributes are large?

% \mysubpara{C2: Information loss because of DP.} 
% Since the attributes of a single record are spread among different parties in the vertical setting, the privacy budget is split among all parties to achieve sufficient protection. 
% Furthermore, there are multiple cross-party attribute combinations to estimate, which would introduce more noise by the nature of DP algorithms. 
% Therefore, the utility loss in the vertical setting would be more significant than in the centralized setting when the same level of DP is achieved.

% \mysubpara{C3: High communication cost.} 
% If a VFL algorithm relies on cryptography protocols (i.e., the inner-product protocol) to secure cross-party computation, it may incur huge communication overhead and can be infeasible when the data volume is large.

Although there are a few works on DP data synthesis under the vertical setting, they still have limitations related to the challenges above.
DistDiffGen~\cite{mohammed2013secure} is a two-party DP data synthesis framework. It falls short of handling C1 and C3 because it relies on a given taxonomy tree requiring strong prior knowledge and is tailored to classification tasks only. 
VertiGAN~\cite{jiang2023distributed} adapts the DP-WGAN approach to vertical setting~\cite{jiang2023distributed}.
However, the GAN-based models are proven to be not suitable for synthesizing tabular data with DP, which indicates that C1 and C2 still hinder its practical application. 
DPLT~\cite{tang2019differentially} utilizes a latent tree model to capture the correlation among cross-party attributions. 
However, its application is limited by C3 because it is designed for datasets with binary attributes and suffers from the huge communication and computation costs incurred by the complicated cryptography protocol.  

To handle the challenges, we propose \textsf{VertiMRF} for generating high-quality synthetic data with differential privacy guarantees in the VFL setting with multiple data parties and a semi-honest central server.
Our key observation is that the central DP data synthesis can achieve great performance in terms of privacy-utility trade-off, and the cross-party statistic estimation is necessary but may unavoidably be less accurate.
Thus, \textsf{VertiMRF} adapts, combines, and balances these two components.
\textsf{VertiMRF} adapts PrivMRF~\cite{cai2021data} to capture and share differentially private intra-party attribute statistic information.
We then design special protocols to let the data parties encode and the server decode the cross-party attribute correlation information.
With both intra-party and cross-party attribute correlation information, the server can reconstruct a global MRF for full-view data synthesis.
Our key contributions assembled in \textsf{VertiMRF} are summarized as follows:
\begin{itemize}[leftmargin = *,topsep=0pt,itemsep=1pt]
    \item We propose a communication efficient and differentially private vertical data synthesis framework \textsf{VertiMRF}. \textsf{VertiMRF} merges a sequence of strategies that allow an untrusted server to construct a global Markov Random Field by merging and balancing differential private encoded information.  
    \item We incorporate a novel Flajolet-Martin (FM) sketch based approach to estimating cross-party multi-attribute marginals. This approach is a key component of \textsf{VertiMRF} to estimate cross-party correlations with relatively low error while protecting privacy. Theoretical privacy guarantee and error analysis are provided. 
    \item We design two critical techniques into \textsf{VertiMRF} to prevent the noise of FM-sketch from obscuring the useful information of attributes with large domain sizes when building the global MRF, including a dimension reduction technique to tune the granularities of attributes while preserving the statistical information and a consistency enforcement technique to maintain the consistency among frequencies of different granularities.
    \item We conduct empirical validation on four real-world datasets. The end-to-end comparison results demonstrate the superiority of our approach to the baseline algorithms. Furthermore, the impact and effectiveness of each component of our approach are validated by ablation studies.
\end{itemize}

\section{Preliminaries}\label{sec:pre}

% \mypara{Vertical DP.}
% As mentioned, in the vertical setting, each party holds a subset of the attributes while sharing the same and aligned records. In a practical view, if one record is removed or altered from one party, then the record's data won't appear in other parties either due to the alignment process, and if one record's attributes are changed in one party, the attributes in other parties should also changed due to the correlation among the attributes.  
% Therefore, despite the vertical setting, we still resort to the same definition of neighboring datasets and record-level privacy as in the centralized setting.
% \begin{definition}[Vertically Differential Privacy]\label{def:VDP}
% Let $D$ and $D'$ be datasets with $d$ attributes $\{A^1,...,A^d\}$. We assume that the attributes are partitioned over $M$ parties $\{\mathcal{P}_1,...,\mathcal{P}_M\}$, where party $\mathcal{P}_i$ holds $d_i$ attributes and $\sum_i{d_i}=d$. Then $D'$ is a neighboring dataset of $D$ if for each party, one user's data is removed from $D$. And a randomized mechanism $\mathcal{M}$ is $(\epsilon,\delta)$-VDP if for all pair of neighboring dataset $D$, $D'$ and all measurable subsets $R\subset Range(\mathcal{M})$, it holds that
% \begin{equation}
%     Pr[\mathcal{M}{(D)}\subseteq R] \leq \exp(\epsilon)Pr[\mathcal{M}{(D')}\subseteq R] +\delta
% \end{equation}
% \end{definition}

\subsection{Differential Privacy}\label{subsec:DP}
Differential privacy (DP) is a rigorous privacy notion that quantifies the privacy loss of algorithms by analyzing the statistical difference between the algorithm outputs on neighboring datasets differing on only one record.

\begin{definition}[Differential Privacy~\cite{dwork2006calibrating}]
A randomized mechanism $\mathcal{M}$ satisfies $(\epsilon,\delta)$-differential privacy if for any neighboring datasets $D$, $D^{\prime} \in \mathcal{D}$ that differ on only one record, their outputs fall in any $R\subset Range(\mathcal{M})$ with probability
$
    Pr[\mathcal{M}(D)\subseteq R] \leq \exp{(\epsilon)}   Pr[\mathcal{M}(D^{\prime})\subseteq R] + \delta .
$
\end{definition}

% DP is widely used in complicated algorithm design due to its properties of composable privacy costs. 
DP is a popular privacy notion because the privacy loss is composable.
Basically, with any two algorithms $f$ and $g$ which satisfy $(\epsilon_1, \delta_1)$-DP and $(\epsilon_2, \delta_2)$-DP respectively, the sequential use of $f\odot g$ on a dataset satisfies $(\epsilon_1+\epsilon_2, \delta_1+\delta_2)$-DP. 
However, such composition is not tight, which may incur huge algorithm utility loss with the overwhelming noises when a DP mechanism is applied repetitively.
To mitigate this issue, \`Renyi Differential Privacy (RDP) is proposed to account for more accurate sequential privacy losses.
% RDP builds an extended privacy model to DP based on the \`Renyi divergence.

\begin{definition}[R{\'e}nyi Differential Privacy~\cite{mironov2017renyi}]
A randomized mechanism $\mathcal{M}: D\rightarrow R$ is said to be $(\lambda, \epsilon)$-R{\'e}nyi Differential Privacy, if for any two neighboring datasets $D$, $D^{\prime}$, it holds that
\begin{align}
    D_\lambda\left(\mathcal{M}(D)\vert \mathcal{M}(D^{\prime})\right) \triangleq \frac{1}{\lambda}\log\mathbb{E}_{ R}\left(\frac{\mathcal{M}(D)\subseteq R}{\mathcal{M}(D^{\prime})\subseteq R}\right)^{\lambda} \leq \epsilon
\end{align}
\end{definition}

RDP can provide a tighter bound of DP when composing a large number of DP mechanisms.
\begin{lemma}[DP Composition based on RDP]\label{lem:com_rdp}
Let $f$ be the composition of $n$ mechanisms that satisfies $\epsilon$-DP, then for each $0<\delta<1$ with $\log{(1/\delta)}\geq n\epsilon^2$, $f$ satisfies $(4\epsilon\sqrt{2n\log{1/\delta}}, \delta)$-DP. 
\end{lemma}
% For any two randomized mechanisms $f: D\rightarrow R_1$ and $g: D\times R_1\rightarrow R_2$ which satisfy $(\lambda, \epsilon_1)$-RDP and $(\lambda, \epsilon_2)$-RDP respectively, the mechanism defined as $(X, Y)$, where $X\sim f(D)$ and $Y\sim g(D)$, satisfies $(\lambda, \epsilon_1 + \epsilon_2)$-RDP.

% \begin{lemma}[RDP to DP]\label{lem:rdp_2_dp}
% If $\mathcal{M}$ is an $(\lambda, \epsilon)$-RDP mechanism, it also satisfies $(\epsilon+\frac{\log{1/\delta}}{\lambda-1},\delta)$-differential privacy for any $0<\delta<1$.
% \end{lemma}

% Based on Lemmas~\ref{lem:com_rdp} and~\ref{lem:rdp_2_dp}, RDP can improve the privacy composition result from order $O(t\epsilon)$ to $O(\sqrt{t}\epsilon)$ when applying a private mechanism $t$ time in a sequence.

\subsection{DP Flajolet-Martin Sketch}\label{subsec:fm}

Flajolet-Martin (FM) Sketch is a probabilistic data structure for multi-set cardinality estimation with DP guarantee. It is constructed by hashing each element in a multi-set to an integer by a hash function $\mathcal{H}$ with a key $\xi$.
The hashed integers are then independent geometric random variables with the parameter $\frac{\gamma}{1+\gamma}$ if $\xi$ is sampled from a large set uniformly. 
Note that, the duplicated elements in the multi-set are mapped to the same integer.
Thus, the cardinality $k$ can be estimated as $k=(1+\gamma)^\alpha$ where $\alpha$ denotes the maximum of the observed integer after hashing.
It has been proved in~\cite{smith2020flajolet} that $(1+\gamma)^\alpha\in \left[\frac{k}{(1+\gamma)}, k(1+\gamma)\right]$ with a reasonable probability.
% As shown, the estimation of $k$ suffers from a multiplicative error of $\gamma+1$. 
The estimation can be improved by repeating the procedure~\cite{smith2020flajolet} multiple times with different hash functions and taking the $1/e$-th quantile of all the maxima as the final estimator. 

The FM sketch-based cardinality estimation is widely used due to its appealing property that the sketch structure is mergeable. 
That is, given two different multi-set $\mathcal{X}_1$ and $\mathcal{X}_2$ and their corresponding FM sketches $\alpha_1$ and $\alpha_2$, then the cardinality of their union $\mathcal{X}_1\cup\mathcal{X}_2$ can be simply estimated as $(1+\gamma)^{\max{(\alpha_1, \alpha_2)}}$.

Based on this and the inclusion-exclusion principle, i.e., $\mathcal{X}_1\cap\mathcal{X}_2=\overline{ \overline{\mathcal{X}_2}\cup \overline{\mathcal{X}_2}}$, the cardinality of the intersection of two multi-sets can also be estimated.

\mypara{\textbf{Differentially private FM-sketch.}}  
As mentioned, estimating the cardinality of a multi-set using FM-sketch involves mapping distinct elements to geometric random variables and selecting the maximum value. However, this process may violate the privacy constraint as it requires the access to the raw data and the maximum value may reveal statistical information about the set. Fortunately, recent studies~\cite{smith2020flajolet, dickens2022order} have demonstrated that FM-sketch can preserve DP under certain conditions. Specifically, if the multi-set contains at least $\frac{1}{e^{\epsilon}-1}$ distinct elements and the maximum of the geometric random variables is lower bounded by $\lceil\log_{1+\gamma}{\frac{1}{1-e^{-\epsilon}}}\rceil$, where $\gamma$ is parameter of the geometric distribution, then the process of selecting the maximum of these random variables ensures $\epsilon$-DP. The privacy guarantee is formalized in lemma~\ref{lemma:dp-fm}, and the DP FM-sketching algorithm is detailed in Algorithm~\ref{alg:dpfm}.
% In particular, Assume $\mathcal{H_{\xi}}: R\times[u]\rightarrow \mathbb{N}^{+}$ is an idealized geometric hash function with hash key $\xi\sim Uniform(R)$ which is unknown to adversaries.  $\mathcal{H_{\xi}}$ can map each element $x\sim [u]$ to an independent geometric random variable with parameter $\left(\frac{\gamma}{1+\gamma}\right)$. Under this condition, if there exists at least $\frac{1}{e^{\epsilon}-1}$ distinct elements in a multi-set and the maximum value of the variables is lower bounded by $\lceil\log_{1+\gamma}{\frac{1}{1-e^{-\epsilon}}}\rceil$, then the operation of taking the maximum of all the random variables preserves $\epsilon$-differential privacy. The privacy guarantee is given in lemma~\ref{lemma:dp-fm} and the details of the DP FM-sketching algorithm are shown in Algorithm~\ref{alg:dpfm}.

\begin{algorithm}[t]
\renewcommand{\algorithmicrequire}{\textbf{Input:}}
\renewcommand{\algorithmicensure}{\textbf{Output:}}
\caption{\textbf{DPFM}}
\label{alg:dpfm}
\begin{algorithmic}[1]
\REQUIRE Multi-set $\mathcal{X} = \{x_1, ...,x_n\}$, domain $[u]$, distribution parameter $\gamma$, privacy budget $\epsilon^{\prime}$, hash key $\xi \sim Uniform(R)$.
\ENSURE DP FM-sketch $\alpha$ for $\mathcal{X}$.
\STATE {$k_p \leftarrow \lceil \frac{1}{e^{\epsilon^{\prime}-1}} \rceil$, $\alpha_{min}\leftarrow \lceil\log_{1+\gamma}{\frac{1}{1-e^{-\epsilon^{\prime}}}}\rceil$}
\STATE {$\alpha_{p}\leftarrow \max\{Y_1, ..., Y_{k_p}\}$ where $Y_i\sim Geometric(\frac{\gamma}{1+\gamma}), \forall i\leq k_p$}
\STATE {$\alpha_{\mathcal{X}} \leftarrow \max\{\mathcal{H}_{\xi}(x_j)\}, \forall x_j\in \mathcal{X}$}.
\RETURN {$\max\left\{\alpha_{\mathcal{X}}, \alpha_{p}, \alpha_{min}\right\}$} 
\end{algorithmic}
\end{algorithm}

\begin{lemma}~\label{lemma:dp-fm}
Let $Y_1,\ldots ,Y_{k+1}$ be independent random variables where each $Y_i\sim geometric(\frac{\gamma}{1+\gamma})$. Let $W_1=\max{\{Y_1,\ldots,Y_{k},b\}}$ and $W_2=\max{\{Y_1,\ldots,Y_{k+1},b\}}$. For any $\epsilon$, if $k\geq\frac{1}{e^{\epsilon}-1}$ and $b\geq \lceil\log_{1+\gamma}{\frac{1}{1-e^{-\epsilon}}}\rceil$, then it holds that $|\log\frac{Pr[W_1 = \mathcal{O}]}{Pr[W_2 = \mathcal{O}]}|\leq \epsilon, \forall \mathcal{O}\in \mathbb{N}^{+}$.
\end{lemma}

\subsection{DP Data Synthesis}\label{subsec:privmrf}
Let $D$ be a set of data tuples $\{x^{(1)},\ldots,x^{(n)}\}$. Each tuple consists of values of a set of attributes $\mathcal{A} = \{A^1, \ldots,A^d\}$.
Each attribute $A^j, \forall j\in[d]$ has domain size $u_j$. 
Without loss of generality, we denote the domain of $A^j$ as $[u_j] \triangleq \{1, \ldots, u_j\}$.
With $M\subset \mathcal{A}$, $x_{M}^{(l)}$ denotes the values of tuple $x^{(l)}$ on an attribute set $M$.
Let $T_{M}$ be the counts of occurrences of all possible value tuples of attributes $M$ in $D$.
That is, $T_{M}$ is a vector of length $\prod_{A^j\in M} u_j$ and each element is defined as 
\begin{align}
    T_{M}[\mathbf{v}] =  \sum_{l\in[n]}\mathbb{I}(x_{M}^{(l)} = \mathbf{v}), \quad \forall \mathbf{v} \in \prod_{A^j\in M} [u_j].
\end{align}
$T_{M}$ is referred as the \textit{contingency histogram} of $M$. 

Data synthesis focuses on generating a dataset $\hat{D}$ given $D$ such that ideally $\forall M \subseteq \mathcal{A}, \hat{T}_M\approx T_M$.
A key challenge of DP data synthesis is to circumvent the curse of dimensionality incurred by a large $d$. 
With the increase of $d$, the error of $T_{\mathcal{A}}$ grows exponentially, as DP noise has to be added to each count of the contingency histogram.
To address this challenge, there have been works~\cite{mckenna2019graphical,mckenna2021winning,zhang2017privbayes,cai2021data,zhang2021privsyn} that propose to utilize low-way marginal distributions to approximate the high-way distribution without losing much correlations among the attributes.
Among these works, PrivMRF~\cite{cai2021data}, utilizing \textit{Markov Random Field} (MRF) to model the attribute correlations, shows the state-of-the-art performance. 

The basic idea of PrivMRF is to select an appropriate set of marginals to construct an MRF, which is then used to approximate the joint distribution of all attributes. 
In particular, PrivMRF consists of four phases:

\begin{itemize}[itemindent=2ex, listparindent=0ex, leftmargin=0ex, labelwidth=0px]
\item \textbf{Phases 1: Generate attribute graph.}
% \item{Phases 1: Generate attribute graph.} 
PrivMRF starts by generating an attribute graph $\mathcal{G}$ through greedily linking up each attribute pair $(A^i, A^j)$ in the descending order of noisy R-scores: 
\begin{equation}\label{equ:R-score}
R(A^i, A^j) = \frac{n}{2}\left\Vert Pr[A^i, A^j]-Pr[A^i]\cdot Pr[A^j]\right\Vert_1+\mathcal{N}(0, \sigma_R^2)
\end{equation}
After that, $\mathcal{G}$ is triangulated to ensure the domain size of the maximal clique not exceeding a threshold.

\item\textbf{Phases 2: Choose candidate marginal set.} PrivMRF samples a set of candidate marginals $\mathcal{U}$ from the cliques of triangulated $\mathcal{G}$ and ensure each marginal $M \in \mathcal{U}$ is $\theta$-useful. That is
$
    \frac{n}{\prod_{\footnotesize{A^i\in M}}u_{i}} \leq \theta\cdot g,
$
where $g$ denotes the expected absolute value of the noise to be injected into each count of $T_{M}$. $\theta$-usefulness ensures that the average count in $T_{M}$ is large enough to tolerate the noise.  
% \zitao{more details here}

\item\textbf{Phases 3: Initialize the marginal set.} From $\mathcal{U}$, PrivMRF selects the most highly correlated marginal for each attribute to constitute an initialized marginal set $\mathcal{S}$, which is used to estimate the parameters $\mathbf{\Theta}$ of the MRF. $\mathbf{\Theta}$ is a real vector where each element corresponds to an entry in a contingency histogram $T_M, \forall M \in \mathcal{S}$. The MRF models the distribution of arbitrary tuple $x$ as:
\begin{align}\label{equ:MRF_distri}
    Pr[x] \propto \prod_{M\in \mathcal{S}} \exp\left(\mathbf{\Theta}_{M}[x_M]\right)
\end{align}
where $\mathbf{\Theta}_{M}$ denotes the sub-vector of $\mathbf{\Theta}$ corresponding to $M$, and $\mathbf{\Theta}_{M}[x_M]$ is the element corresponding to $x_M$. Based on the estimated $\mathbf{\Theta}$, any marginal $M^{\prime}$ can be inferred by the MRF as 
\begin{align}
Pr[y] = \sum_{x, x_{M^{\prime}} = y} Pr[x], \quad \forall y \in \prod_{A^j\in M^{\prime}}u_j.
\end{align}

\item\textbf{Phases 4: Refine the marginal set.} PrivMRF proceeds to refine the marginal set $\mathcal{S}$ by inserting marginals that cannot be accurately inferred by the MRF and iteratively refine the estimation of MRF.
\end{itemize}

\section{Differentially Private Vertical Data Synthesis}\label{sec:methods}
% In this section, we first formulate the problem of differentially private data synthesis in the vertical federated learning setting, and then provide the overview of our solution.
% In this paper, we focus on DP data synthesis in the challenging vertical setting.
We provide the problem definition of DP data synthesis in the vertical setting and an overview of our solution in this section.

% This section formulates the problem and presents the overall framework of our solution \textsf{VertiMRF}.

\subsection{Problem Definition}\label{subsec:problem_def}

We consider a system constituted by $m$ data parties and an untrusted central server orchestrating the overall process.
Each data party $\mathcal{P}_{i}, \forall i \in [m]$, possesses users' data $D_i = \{x^{(1)}_{\mathcal{A}_i}, \ldots, x^{(n)}_{\mathcal{A}_i} \}$ with a subset of attributes $\mathcal{A}_i\subset \mathcal{A}$.
We assume that the user's data has been aligned across these $m$ data parties by some record ID (e.g., social security number and phone number) with some private set intersection method~\cite{kolesnikov2016efficient,dong2013private,chen2017fast}.
That is, $x^{(l)}_{\mathcal{A}_i}$ and $x^{(l)}_{\mathcal{A}_j}$ are data tuples of a same individual $l$ but on different attributes.
The aligned data is a common setting with the vertical tasks~\cite{chen2021secureboost+, hardy2017private, liu2019communication,xie2022improving,zhang2020additively}.
Virtually speaking, there is a global dataset $D = (D_1 | \ldots | D_m)$ with attributes $\mathcal{A} = \cup_{i\in [m]} \mathcal{A}_i$ if all data parties' data can be merged.

\mypara{Adversary model.}
We consider both the adversaries within and outside the system. 
By the adversary within the system, we consider the central server to be honest but curious, which would execute the protocol honestly but try his best to infer the private information of the input dataset. 
However, we assume that none of the data parties is interested in colluding with the central server because privacy regulations prevent data parties from doing so. 
We consider the adversary outside the system as all the third-party data analysts who aim to infer some private information of the input dataset from the synthetic dataset and the intermediate results carried out in the communication between data parties and the server.

\mypara{Our goal.}
\textit{Our work aims to generate a collection of synthetic data $\hat{D}$ with attributes $\mathcal{A}$, which follows the data distribution as the virtual global dataset $D$ as closely as possible while protecting the privacy information. }
Specifically, we employ DP to ensure a high probability, controlled by the privacy budget parameter $\epsilon$, that no adversary can infer based on the synthesized dataset $\hat{D}$ whether any individual's data is used as the input of the data synthesis algorithm.

\subsection{Overview of Our Solution}\label{subsec:solution}

\begin{figure*}[htbp]
  \centering
  \begin{minipage}[t]{0.65\textwidth}
    \centering
    \includegraphics[width=\linewidth]{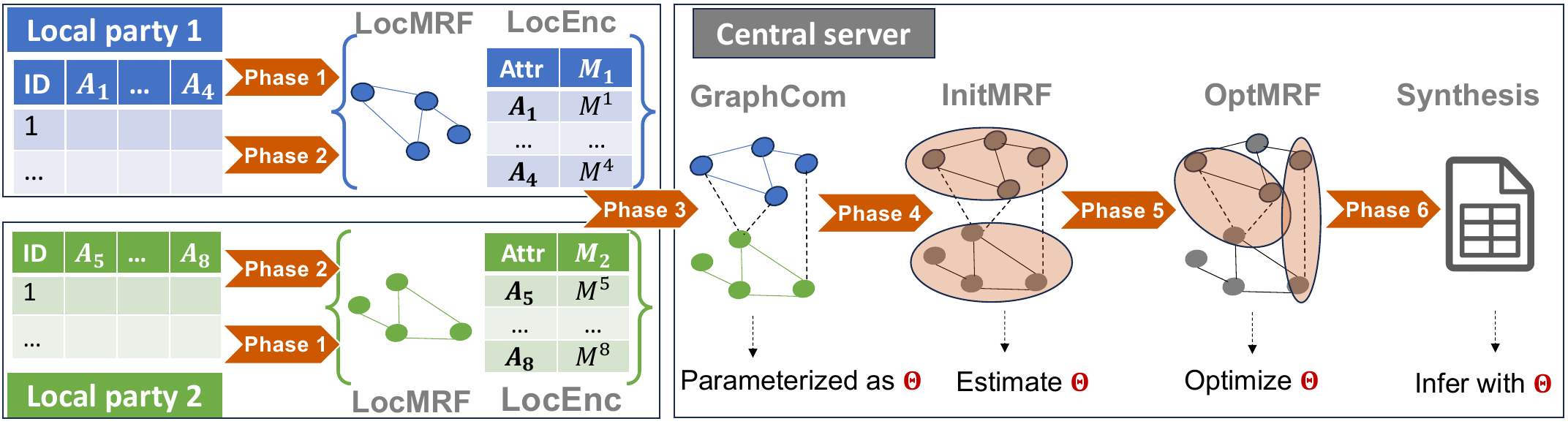}
    \vspace{-0.2in}
    \caption{Workflow of \textsf{VertiMRF}}
    \label{fig:workflow}
  \end{minipage}
  \hfill
  \begin{minipage}[t]{0.28\textwidth}
    \centering
    \includegraphics[width=\linewidth]{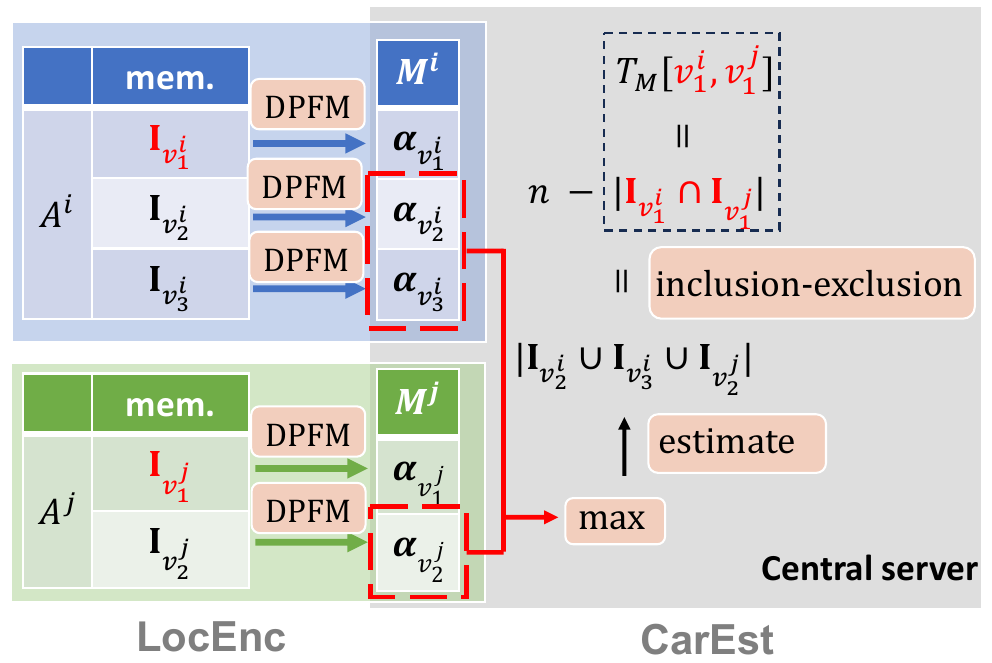}
    \vspace{-0.2in}
    \caption{\texttt{LocEnc} and \texttt{CarEst}. \label{fig:locenc}}
    \label{fig:membership}
  \end{minipage}
\end{figure*}
% \begin{figure*}[htbp]
%   \centering
%   \begin{minipage}[t]{0.68\textwidth}
%     \centering
%     \includegraphics[width=\linewidth]{figs/workflow.pdf}
%     \caption{Workflow of \textsf{VertiMRF}}
%     \label{fig:workflow}
%   \end{minipage}
%   \hfill
%   \begin{minipage}[t]{0.28\textwidth}
% 	\centering
% 	\includegraphics[width=\textwidth]{figs/locenc.pdf}
%     \caption{\texttt{LocEnc} and \texttt{CarEst}. \label{fig:locenc}}
% \end{minipage}
% \end{figure*}

To address the problem defined above, we propose \textsf{VertiMRF}, a novel differentially private data synthesis approach. Figure~\ref{fig:workflow} and Algorithm~\ref{alg:overall_workflow} visualize the workflow of \textsf{VertiMRF}, which can be divided into the following six phases:

\begin{itemize}[itemindent=2ex, listparindent=0ex, leftmargin=0ex, labelwidth=0px]
\item\textbf{Phase 1:} Each party $\mathcal{P}_i$ constructs a local Markov Random Field $\texttt{MRF}_i$ to capture the correlation among local attributes $\mathcal{A}_i$. Besides,  $\mathcal{P}_i$ preserves the inner results, including the local attribute graph $\mathcal{G}_i$ and the marginal set $\mathcal{S}_i$ (sub-procedure \texttt{LocMRF}).

\item\textbf{Phase 2:} Each party $\mathcal{P}_i$ encodes local dataset with attributes $\mathcal{A}_i$ via differentially private FM sketch. Both the codes $\mathcal{M}_i$ 
% = \{\mathcal{M}^j\vert A^j\in \mathcal{A}_i\}$ 
and $\{\texttt{MRF}_i, \mathcal{G}_i, \mathcal{S}_i\}$ are sent to the central server (sub-procedure \texttt{LocEnc}).

\item\textbf{Phase 3:} The server generates a global attribute graph $\mathcal{G}$ by combining received disjoint local attribute graphs $\{\mathcal{G}_i\vert i \in [m]\}$. In the combining, server links up cross-party attribute pairs with higher R-scores estimated over the encoded attributes $\{M_i, i\in[m]\}$. The generated $\mathcal{G}$ is then triangulated (sub-procedure \texttt{GraphCom}).

\item\textbf{Phase 4:} The server initializes a marginal set $\mathcal{S}$ by taking the union of the received local marginal set. Based on $\mathcal{S}$, the parameter $\Theta$ of the global MRF is initialized with each contingency histogram $T_M, \forall M \in \mathcal{S}$ inferred from the received local \texttt{MRF}s (sub-procedure \texttt{InitMRF}). 

\item\textbf{Phase 5:} The server selects a set of cross-party marginals $\mathcal{S}^{c}$ from the cliques of $\mathcal{G}$. Based on the $\mathcal{S}^{c}$, $\Theta$ is further optimized. In the optimization, each contingency histogram $T_M, \forall M \in \mathcal{S}^{c}$ is estimated over the encoded attributes (sub-procedure \texttt{OptMRF}).

\item\textbf{Phase 6:} The server generates synthetic data by sampling from the data distribution approximated by the global MRF.
\end{itemize}

% 调整：1. Phase 1 仅调整参数, 在参数的选择中xxx。 2. section 3 跨域相关，basic building block技术. 3. 怎么用section 4 的技术来解决section 5的 跨域mrf goujian
In what follows, we show the solution for \textbf{Phase 1-2} in Section~\ref{sec:mar_est} which describes the DP information sharing approaches of each local party. Then we describe \textbf{Phase 3-6} in Section~\ref{sec:server_mrf}, presenting how to use the shared DP information to construct a global MRF.

% the usage of \texttt{LocEnc} and \texttt{CarEst} in constructing a global MRF.

% . Specifically,  we present \texttt{LocEnc} and \texttt{CarEst}, two basic building-block techniques for realizing the cross-party marginal estimation. 

% \subsubsection{\textbf{\texttt{LocEnc}}} Each party $\mathcal{P}_i$ applies \texttt{LocEnc} to privately encode the local attributes in $\mathcal{A}_i$. Such encoding can help estimate the contingency histogram of cross-party marginals subsequently in \textbf{Phase 3} and \textbf{Phase 5}. The details of  \texttt{LocEnc} are then described in Section~\ref{sec:mar_est}.

\begin{algorithm}[t]
\renewcommand{\algorithmicrequire}{\textbf{Input:}}
\renewcommand{\algorithmicensure}{\textbf{Output:}}
\caption{\textbf{\textsf{VertiMRF}}}
\label{alg:overall_workflow}
\begin{algorithmic}[1]
\REQUIRE The partitioned dataset $D = \{D_i, i \in [m]\}$, domain $([u_1]\times...\times [u_d])$, maximal clique size $\tau$, total privacy budget $(\epsilon, \delta)$ is divided as $\epsilon_0 = \frac{\epsilon}{2m}$, $\delta_0 = \frac{\delta}{2m}$, $\epsilon_1 = \frac{\epsilon}{2}$, $\delta_1 = \frac{\delta}{2}$.
\ENSURE Synthesized data $\hat{D}$.
\STATE {Each local party $\mathcal{P}_i$:\\
(a). constructs local \small{MRF}: $\{\texttt{MRF}_i,\mathcal{G}_i,\mathcal{S}_{i}\}\leftarrow\texttt{LocMRF}(D_i, \tau, \epsilon_0, \delta_0)$}.
\STATE {Each local party $\mathcal{P}_i$:\\
(a). encodes local attributes: $\mathcal{M}_i \leftarrow \texttt{LocEnc}(D_i,\mathcal{A}_i, \epsilon_1, \delta_1)$.\\
(b). sends $\mathcal{M}_i$ and $\{\texttt{MRF}_i,\mathcal{G}_i,\mathcal{S}_{i}\}$ to server.}
\STATE {Central server:\\
(a). generates global graph: \small{$\mathcal{G}\leftarrow$}\texttt{GraphCom}($\small{\{\mathcal{G}_i, \mathcal{M}_i\vert i \in [m]\}}$).}
\STATE {Central server:\\
(a). initializes marginal set: $\mathcal{S} \leftarrow \bigcup_{i=1}^{m}\{\mathcal{S}_{i}\}$.\\
(b). initializes parameter $\Theta$ of the global MRF based on $\mathcal{S}$}. 
\STATE {Central server:\\
(a). selects cross-party marginals $\mathcal{S}^c$ from triangulated $\mathcal{G}$.\\
% (b). refines marginal set $\mathcal{S}$: $\mathcal{S} \leftarrow \mathcal{S} \bigcup\mathcal{S}^c$.\\
(b). optimizes $\Theta$ based on $\mathcal{S}^c$.
}
\STATE {Central server:\\
(a). samples $\hat{D}$ based on the optimized global MRF.}
\end{algorithmic}
\end{algorithm}

\section{differentially private information sharing}\label{sec:mar_est}
% In this section, we first demonstrate the reuse of PrivMRF for \textbf{Phase 1} and then present \texttt{LocEnc} (\textbf{Phase 2}) and \texttt{CarEst} which are two basic building blocks to estimate the cross-party marginals. We show that each party conveys effective information with DP guarantee.

Based on our security setting and DP's resistance to post-processing, the key to satisfying privacy protection is to ensure differential privacy guarantee for all the information shared from local parties, which are the outputs of \texttt{LocMRF} in \textbf{Phase 1} and \texttt{LocEnc} \textbf{Phase 2} (in brackets Figure~\ref{fig:workflow}).
Thus, we introduce the algorithms for \texttt{LocMRF} and  \texttt{LocEnc} (together with its closely paired \texttt{CarEst}), providing bases of the following synthesis steps.

% introduce a novel \texttt{LocEnc} approach for \textbf{Phase 2}. We show that each party conveys important information with provable  DP guarantee. Additionally, we describe \texttt{CarEst}, a novel approach to estimating the cross-party marginals attached to \texttt{LocEnc}.

% In this section, we first present an LDP version of \texttt{LocEnc} and \texttt{CarEst}. Then we attach more importance on our FM sketch based solutions. 

\subsection{Local PrivMRF in Phase 1}\label{subsec:reuse_privmrf}

Each local party $\mathcal{P}_i$ directly applies the PrivMRF approach to construct $\texttt{MRF}_i$. As shown in Section~\ref{subsec:privmrf}, there would be inner results generated when constructing $\texttt{MRF}_i$, including the attribute graph $\mathcal{G}_i$ and the refined marginal set $\mathcal{S}_i$. Apart from $\texttt{MRF}_i$, both $\mathcal{G}_i$ and $\mathcal{S}_i$ should also be preserved and sent to the central server. 
Notably, because the maximal clique size for the global MRF is always limited to control the complexity of the attribute graph, the maximal clique size of each local MRF should also be limited. The maximal local clique size for each $\texttt{MRF}_i$ is set as $\tau^{\prime} \leq \frac{\tau}{m\cdot \bar{u}^2}$, with $\bar{u} = \frac{\sum_{j}u_j}{d}$ and $\tau$ is threshold of the clique size for global MRF. The constructed $\texttt{MRF}_i$ captures the correlations among the local attributes. 

\subsection{Frequency Oracle as a Baseline}

Frequency oracle (FO) protocols provide DP protection by randomizing each user's data and allowing the frequency estimation of values in the original domain. 
In this case, we use the widely known FO protocol, called the Generalized Random Response technique (GRR)~\cite{wang2019answering}, to implement a baseline for \texttt{LocEnc}.

\mypara{FO-based \texttt{LocEnc}.}
Each party $\mathcal{P}_i$ employs GRR to encode the local dataset. Specifically, for each local attribute $A^j \in \mathcal{A}_{i}$, a value $v_{(j)}$ is perturbed to an arbitrary $v_{(j)}^{\prime} \in [u_j]$ with probabilities:
\begin{equation}\label{equ:grr}
Pr\left[v_{(j)}^{\prime} = v\right] = \left\{
\begin{array}{rcl}
\frac{e^{\epsilon^{\prime}}}{e^{\epsilon^{\prime}}+u_{j}-1}, & v = v_{(j)}\\
\frac{1}{e^{\epsilon^{\prime}}+u_{j}-1}, & v\neq v_{(j)}\\
\end{array}
\right.
\end{equation}
Here, $\epsilon^{\prime}$ denotes LDP level preserved by the GRR technique for each attribute. 
After applying GRR to each user's data value in the local dataset $D_i$, we obtain a perturbed version $\widetilde{D}_i$. The partition of $\widetilde{D}_i$ restricted to $A^j$, denoted as $\mathcal{M}^j$, is taken as the encoded attribute of $A^j$. Subsequently, the encoded local attributes $\mathcal{M}_i = \{\mathcal{M}^j, \forall A^j\in \mathcal{A}_i \}$ are reported to the central server.

\mypara{FO-based \texttt{CarEst}.}
After receiving the reported encoded attributes $\mathcal{M}=\{\mathcal{M}^j\vert j\in [d]\}$, the central server can estimate the contingency histogram of any arbitrary marginal.
For an $l$-way marginal $M = (A^1,\ldots, A^l)$, we obtain a noisy contingency histogram $\hat{T}_{M}$ by counting the occurrences of each value tuple $\mathbf{v} = (v_{(1)},\ldots, v_{(l)})\in \prod_{i=1}^{l} [u_i]$ from $\mathcal{M}$. 
However, relying solely on this estimation can introduce considerable bias. To mitigate this issue, a commonly employed technique is to utilize a transition probability matrix $P$ to overcome the bias, which would then produce an unbiased estimate.

As shown in Equation~(\ref{equ:grr}), different attributes are encoded independently in \texttt{LocEnc} procedure. 
So each value tuple $\mathbf{v} = (v_{(1)},\ldots, v_{(l)})$ can be encoded as any arbitrary $\mathbf{v}^{\prime} = (v_{(1)}^{\prime}, \ldots, v_{(l)}^{\prime})$ with probability  
$Pr\left[\mathbf{v}\rightarrow \mathbf{v}^{\prime}\right] = \prod_{A^i\in M}Pr\left[v_{(i)}\rightarrow v_{(i)}^{\prime}\right].$
Since there are $\prod_{A^i\in M}u_i$ possible values for $\mathbf{v}$ in total, we can construct a $\left(\prod_{A^i\in M}u_i\right)\times \left(\prod_{A^i\in M}u_i\right)$-dimensional probability matrix $P$ to establish the transition relationship between $T_M$ and the noisy $\hat{T}_M$. That is $P\cdot T_M = \mathbb{E}[\hat{T}_M]$,
where the expectation accounts for the randomness of GRR. 
Therefore, $T_M$ can be estimated as $P^{-1}\cdot\hat{T}_M$, where the existence of $P^{-1}$ is guaranteed by the positive definite property of the matrix. Furthermore, it can be shown that $$\mathbb{E}[P^{-1}\cdot\hat{T}_M] = P^{-1}\cdot \mathbb{E}[\hat{T}_M]=P^{-1}\cdot P\cdot T_M = T_M.$$ This implies that $P^{-1}\cdot\hat{T}_M$ is an unbiased estimator of $T_M$.

\begin{theorem}[\textbf{Privacy $\&$ Error Analysis}]\label{the:rr-error}
Given a marginal $M=(A^{1},\ldots ,A^{l})$, if each attribute $A_i \in M$ is encoded with $\epsilon^{\prime}$-LDP following the rule shown in Equation~(\ref{equ:grr}), then the FO-based  $\texttt{LocEnc}$ preserves $\left(\min \left\{d\epsilon^{\prime}/2, 2\epsilon^{\prime}\sqrt{2d\log(1/\delta)}\right\},\delta\right)$-DP, $\forall \delta <1$, where $\delta = 0$ when the minimum taking $d\epsilon^{\prime}/2$. The FO-based $\texttt{CarEst}$ gives unbiased estimation for each count in $T_M$ with variance $O(n/{\epsilon^{\prime}}^{2l})$.
\end{theorem}
\begin{proof}
% The details are shown in the Appendix. 
The details are shown in Appendix~\ref{proof:the3}.
\end{proof}

\subsection{Sketch-based \texttt{LocEnc} and \texttt{CarEst}}

As explained in Section~\ref{subsec:fm}, FM sketch can be used to estimate the cardinality of a multi-set. And the estimation process can easily satisfy DP by incorporating phantom elements and bounding the maximum value of the hashed geometric variables.
Building on this idea, we design our sketch-based \texttt{LocEnc} and \texttt{CarEst}.  Figure~\ref{fig:locenc} visualizes the rationale of both sketch-based \texttt{LocEnc} and \texttt{CarEst}.

% \textcolor{blue}{
% It is worthy to mention a notable study~\cite{li2023differentially} that utilizes DP FM sketch in a DP vertical clustering task for encoding membership information of local clusters and estimate the cardinality of cross-party clusters' intersection. In that work, each party only needs to encode a single set of local clusters. In contrast, our approach requires each party to encode all local attributes to facilitate arbitrary marginal estimation.
Note that our sketch-based estimation is inspired by \cite{li2023differentially} that utilizes DP FM sketch to encode membership information of local clusters and estimate the cardinality of cross-party clusters' intersection.
However, only one attribute is shared by a data party in the clustering task (i.e., which cluster an individual is clustered to), while each data party possesses data with multiple attributes in data synthesis.
If we want the server to estimate any cross-party attribute combination with one communication round, using different hash keys for each combination is infeasible as each hash key introduces additional privacy loss.
We solve the challenge by answering the following questions: whether encoding multi-dimensional memberships with-in the same party with the same key still provides privacy protection and whether this approach can provide satisfying privacy-utility trade-offs.
% }

% for realizing cross-party . Differently, our problem is more challenging since the  

\mypara{Sketch-based \texttt{LocEnc}.}
Each data party $\mathcal{P}_i$ encodes the membership information of local attribute $A^j$ using an FM sketch. This information, denoted as $\left\{\mathbb{I}_{v_1^j},\ldots, \mathbb{I}_{v_{u_j}^j}\right\}$, consists of $u_j$ ID sets. Each set $\mathbb{I}_{v_i^j}$ contains the IDs of records $x$ with $x_{(A^j)} = v_i^j$. An example illustrating the membership information of attributes is provided in \textit{Example 1}. The sketch-based \texttt{LocEnc} involves two main procedures: the generation of hash keys and the generation of sketches.

\noindent \textit{Example 1: let $D$ be a dataset containing records with attributes \underline{Gender}, \underline{Age} and \underline{Hobby} as in the following table. Then, for attribute \underline{Gender}, the membership information 
is $\{\mathbb{I}_{male}, \mathbb{I}_{female}\}$, where $\mathbb{I}_{male} = \{1\}$ and $\mathbb{I}_{female} = \{2,3\}$.}

\begin{table}[!h]
\centering
% \vspace{-0.3cm}
% \caption{Example dataset}
\resizebox{0.45\columnwidth}{!}{%
\vspace{-0.35cm}
% \label{tab:example}
\begin{tabular}{@{}cccc@{}}
\toprule
Index & Gender & Age  & Hobby \\ \midrule
1    & male    & 20-30  & cook          \\ \midrule
2    & female   & 20-30   & basketball      \\ \midrule
3   & female   & 10-20  & cook          \\ \bottomrule
\end{tabular}
}
\vspace{-0.3cm}
\end{table}

Due to the privacy concern, the hash keys should be collaboratively generated by the data parties and kept unknown to the central server. There are multiple secure multi-party computation (SMC) protocols can be applied to achieve this, such as the Diffie-Hellman protocol~\cite{krawczyk2005hmqv}, which allows multi parties to negotiate a random number securely even if the central server is semi-honest~\cite{krawczyk2005hmqv}.
% One approach is that, each party $\mathcal{P}_i$ independently generates a random number $r_i$ and exchanges it with other parties through secure peer-to-peer (P2P) channels. Upon receiving the numbers from other parties, each party can compute the hash key as $\xi = (\sum_i{r_i})\mod R$, where $R$ is the range of the hash key. If secure P2P channels are not available, an alternative is to use the Diffie-Hellman protocol~\cite{krawczyk2005hmqv} to generate a hash key, with all intermediate data exchanged through the central server. This process remains secure even if the central server is semi-honest, as assumed in~\cite{krawczyk2005hmqv}. 

Next, each party encodes the membership information of local attributes using \texttt{DPFM} (Algorithm~\ref{alg:dpfm}) algorithm with the generated hash key $\xi$. Specifically, for the membership information $\{\mathbb{I}_{v_1^j},\ldots, \mathbb{I}_{v_{u_j}^j}\}$ of attribute $A^j\in \mathcal{A}_i$, party $\mathcal{P}_i$ applies the \texttt{DPFM} algorithm  to each $\mathbb{I}_{v_i^j}$ with a given privacy budget $\epsilon^{\prime}$. This generates a DP FM sketch tuple $(\alpha_{v_1^j},\ldots, \alpha_{v_{u_j}^j})$ for $A^j\in\mathcal{A}_i$. Considering all local attributes, party $\mathcal{P}_i$ composes a tuple set $\{(\alpha_{v_1^j},\ldots, \alpha_{v_{u_j}^j})\vert A^j \in \mathcal{A}_i\}$. To enhance utility, this process is repeated $t$ times, and party $\mathcal{P}_i$ sends $t$ tuple sets $\left\{\mathcal{M}_{i}^{(h)} \triangleq\left\{\left(\alpha_{v_1^j}^{(h)},\ldots,\alpha_{v_{u_j}^j}^{(h)}\right)\middle|\ A^j\in \mathcal{A}_i\right\}\middle|\ h\in [t]\right\}$ to the central server. The details of the sketch-based \texttt{LocEnc} method are presented in Algorithm~\ref{alg:fm_LocEnc}.

\begin{algorithm}[t]
\renewcommand{\algorithmicrequire}{\textbf{Input:}}
\renewcommand{\algorithmicensure}{\textbf{Output:}}
\caption{\textbf{Sketch-based \texttt{LocEnc}}}
\label{alg:fm_LocEnc}
\begin{algorithmic}[1]
\REQUIRE $\mathcal{P}_i$'s local dataset $D_i$, attribute $A^j\in \mathcal{A}_i$, distribution parameter $\gamma$, privacy budget $(\epsilon, \delta)$, domain $\left([u_1]\times\ldots\times [u_l]\right)$.
\ENSURE $\mathcal{P}_i$'s sketch set $\mathcal{M}_{i}$.
\STATE {Data parties generate $t$ hash keys $\{\xi_1,\ldots, \xi_t\}$ collaboratively.}
\STATE {$\epsilon^{\prime} = \frac{\epsilon}{4\sqrt{td\log(1/\delta)}}$}
\FOR{$h \in [t]$}
\FOR{each attribute $A^j\in \mathcal{A}_i$}
\FOR{$l \in [u_j]$}
\STATE {$\alpha_{v_{l}^{j}}^{(h)} = \texttt{DPFM}(D_i, u_j, \gamma, \epsilon^{\prime}, \xi_h)$}
\ENDFOR
\ENDFOR
\STATE{$\mathcal{M}_i^{(h)} \leftarrow \left\{\left(\alpha_{v_1^j}^{(h)},\ldots,\alpha_{v_{u_j}^j}^{(h)}\right)\vert A^j\in \mathcal{A}_i\right\}$}
\ENDFOR
\RETURN{$\{\mathcal{M}_i^{(h)}\vert h \in [t]\}$}
\end{algorithmic}
\end{algorithm}

\mypara{Sketch-based \texttt{CarEst}.}
As mentioned in Section~\ref{subsec:fm}, the FM sketch enables us to estimate the cardinality of the intersection of multiple sets using the inclusion-exclusion principle. This property can be extended to the DP FM sketch, allowing the central server to estimate the contingency histogram of a marginal. The details of this estimation process are presented in Algorithm~\ref{alg:fm_CarEst}.

After receiving all the sketches from data parties, the central server aggregates them into $t$ sets of sketch tuples $\left\{\mathcal{M}^{(h)}\triangleq\left\{\left(\alpha^{(h)}_{v_1^j},\ldots, \alpha^{(h)}_{v_{u_j}^j}\right)\vert j\in [d]\right\}\vert h \in [t]\right\}$. For each $l$-way marginal $M = (A^1,\ldots, A^l)$, the estimation of the contingency histogram $T_M$ involves estimating the cardinality of the intersection set $\bigcap_{i=1}^{l}{\mathbb{I}_{v_{(i)}}}$ for each $(v_{(1)},\ldots, v_{(l)})\in \prod_{i=1}^{l} [u_i]$. 
Here, $\mathbb{I}_{v_{(i)}}$ represents the membership information of attribute $A^i$ with value $v_{(i)}$. 
Using the inclusion-exclusion principle (i.e., $\bigcap_{i=1}^{l}{\mathbb{I}_{v_{(i)}}}=\overline{\bigcup_{i=1}^{l}{\overline{\mathbb{I}_{v_{(i)}}}}}$), 
% \begin{align}
% \bigcap_{i=1}^{l}{\mathbb{I}_{v_{(i)}}}=\overline{\bigcup_{i=1}^{l}{\overline{\mathbb{I}_{v_{(i)}}}}}
% \end{align}
the cardinality of $\bigcap_{i=1}^{l}{\mathbb{I}_{v_{(i)}}}$ can be determined by calculating the cardinality of $\overline{\bigcup_{i=1}^{l}{\overline{\mathbb{I}_{v_{(i)}}}}}$, where $\overline{\mathcal{X}}$ denotes the complementary set of $\mathcal{X}$. 
Thus, estimating the cardinality of an intersection is transformed into estimating the cardinality of the complementary set of a union. The basic approach to estimate $\left\vert\overline{\bigcup_{i=1}^{l}{\overline{\mathbb{I}_{v_{(i)}}}}}\right\vert$ is as follows: first, estimate $\left\vert\bigcup_{i=1}^{l}{\overline{\mathbb{I}_{v_{(i)}}}}\right\vert$ using the mergeable property of sketches, and then subtract this estimate from a DP sanitized data number $\hat{n}$.

For each $\mathcal{M}^{(h)}$ among all $t$ sketch sets, the sketch of $\overline{\mathbb{I}_{v_{(i)}}}$ can be estimated by $\max\left\{\alpha^{(h)}_{v^i_j}\vert j \in [u_i], v^i_l\neq v_{(i)}\right\}$. Here, $\alpha^{(h)}_{v^i_j}$ represents the sketch corresponding to attribute $A^i$ with value $v^i_j$. Furthermore, the sketch of $\bigcup_{i=1}^{j}{\overline{\mathbb{I}_{v_{(i)}}}}$ can be estimated by $\max\left\{\max\left\{\alpha^{(h)}_{v^i_j}\vert j \in [u_i], v^i_j\neq v_{(i)}\right\}\vert A^i \in M\right\}.$ After obtaining $t$ estimates of the sketch of $\bigcup_{i=1}^{j}{\overline{\mathbb{I}_{v_{(i)}}}}$, a more stable and accurate estimate $\alpha$ can be obtained by taking the harmonic mean. Furthermore, since the above sketch estimation process involves $\max$ operations on $\sum_{i=1}^{l}{\left(u_i-1\right)}$ sketches, each of which introduces $k_p$ phantom elements as shown in Algorithm~\ref{alg:dpfm}, there should be $\left(\sum_{i=1}^{l}{\left(u_i-1\right)}\right)\cdot k_p$ phantom elements taken into account in total.  By subtracting those phantom elements, $\left\vert\bigcup_{i=1}^{l}{\overline{\mathbb{I}_{v_{(i)}}}}\right\vert$ can be estimated by $(1+\gamma)^{\alpha}-\left(\sum_{i=1}^{l}{\left(u_i-1\right)}\right)\cdot k_p$. 
Finally, the cardinality of $\bigcap_{i=1}^{l}{\mathbb{I}_{v_{(i)}}}$ can be obtained by subtracting the estimated $\left\vert\bigcup_{i=1}^{l}{\overline{\mathbb{I}_{v_{(i)}}}}\right\vert$ from a DP sanitized data number $\hat{n}$ and ensuring the non-negativity.

Notice that \texttt{CarEst} operates on sketches that are generated with privacy guarantees. Therefore, \texttt{CarEst} does not consume additional privacy budget due to the post-process property of DP.

\begin{algorithm}[t]
\renewcommand{\algorithmicrequire}{\textbf{Input:}}
\renewcommand{\algorithmicensure}{\textbf{Output:}}
\caption{\textbf{Sketch-based \texttt{CarEst}}}
\label{alg:fm_CarEst}
\begin{algorithmic}[1]
\REQUIRE Marginal $M = (A^1,\ldots, A^l)$, domain $\left([u_1]\times\ldots\times [u_l]\right)$, sketch sets $\left\{\mathcal{M}^{(h)}\triangleq\left\{\left(\alpha^{(h)}_{v_1^j},\ldots, \alpha^{(h)}_{v_{u_j}^j}\right)\vert j\in [d]\right\}\vert h \in [t]\right\}$, privacy budget $(\epsilon, \delta)$, distribution parameter $\gamma$, noisy data number $\hat{n}$.
\ENSURE Estimated contingency histogram $T_M$.
\STATE {$\mathcal{T} \leftarrow \mathbf{0}^{t\times(u_1\times\ldots\times u_l)}, T_M \leftarrow \mathbf{0}^{(u_1\times\ldots\times u_l)}$}
\STATE {$\epsilon^{\prime} = \frac{\epsilon}{4\sqrt{td\log(1/\delta)}}$, $k_p = \lceil \frac{1}{e^{\epsilon^{\prime}-1}} \rceil$}
\FORALL{$h \in [t]$, $\left(v_{(1)},\ldots, v_{(l)}\right)\in \left([u_1]\times\ldots\times [u_l]\right)$}
\STATE{\footnotesize{$\mathcal{T}[h, (v_{(1)},\ldots, v_{(l)})] = \max\left\{\max\left\{\alpha^{(h)}_{v^i_l}\vert l \in [u_i], v^i_l\neq v_{(i)}\right\}\vert A^i \in M\right\}$}}
\ENDFOR
\FORALL {$(v_{(1)},\ldots, v_{(l)})\in \left(u_1\times\ldots\times u_l\right)$}
\STATE {$\alpha = \texttt{HarmonicMean}\left(\mathcal{T}[:, (v_{(1)},\ldots, v_{(l)})]\right)$}
\STATE {$T_M[(v_{(1)},\ldots, v_{(l)})] = (1+\gamma)^{\alpha} - \sum_{i=1}^{l}{\left(u_i-1\right)}\cdot k_p$}
\STATE {$T_M[(v_{(1)},\ldots, v_{(l)})] = \max\left\{\hat{n} - T_M[(v_{(1)},\ldots, v_{(l)})], 0\right\}$}
\ENDFOR
\RETURN {$T_M$}
\end{algorithmic}
\end{algorithm}

\begin{theorem}[\textbf{Privacy Analysis}]\label{the:fm_priv}
Suppose the FM sketch $\alpha_{v^i_j}^{(h)}$ for value $v^i_{j}$ of attribute $A^i$ is generated with $\epsilon^{\prime}$-DP in the $h$-th run. Then, the sketch-based \texttt{LocEnc} method in Algorithm~\ref{alg:fm_LocEnc} guarantees $(4\epsilon^{\prime}\sqrt{td\log(1/\delta)}, \delta)$-DP for all $\delta<1$.
\end{theorem}

\begin{theorem}[\textbf{Error Analysis}]\label{the:fm_error}
Let $M = \{A^1, \ldots, A^l\}$ be an $l$-way marginal. Suppose $\hat{T}_M$ is the contingency histogram of $M$ estimated using Algorithm~\ref{alg:fm_CarEst} with privacy parameter $(\epsilon, \delta)$ and distribution parameter $\gamma$. For each $\mathbf{v} \in \prod_{i=1}^{l}[u_i]$, the following inequality holds:
\begin{align}
\frac{\vert \hat{T}_M[\mathbf{v}] - T_M[\mathbf{v}]\vert}{T_M[\mathbf{v}]}\leq
\gamma\cdot(\frac{n}{T_M[\mathbf{v}]} - 1) + \frac{\hat{N} + C}{T_M[\mathbf{v}]},
\end{align}
with a probability of at least $1-\beta$. Here, $\hat{N}$ represents the Laplacian noise added to the data number $n$, and $C = O(\frac{\log^{1/2}(1/\delta)\log^{1/4}(1/\beta)}{\epsilon^{\prime}})$.
\end{theorem}

Due to the space limitation, the proofs are shown in the Appendix.
As shown in Theorem~\ref{the:fm_error}, the relative error tends to be larger when the proportion of $T_M[\mathbf{v}]$ in $n$ decreases and when the count $T_M[\mathbf{v}]$ decreases. Meanwhile, a stronger privacy level, represented by a smaller value of $\epsilon$, can indeed lead to a larger relative error, as indicated by the term $C$ in the theorem. 

\subsection{Privacy and Communication Cost}\label{subsec:analysis}

\mypara{Overall Privacy Analysis.}  As shown in the workflow of \textsf{VertiMRF} in Algorithm~\ref{alg:overall_workflow}, \texttt{LocMRF} on all $m$ local parties consumes $(m\cdot \frac{\epsilon}{2m}, m\cdot \frac{\delta}{2m})$-DP.  As stated in Theorem~\ref{the:fm_priv}, the remaining $(\frac{\epsilon}{2}, \frac{\delta}{2})$-DP is allocated for encoding the $d$ attributes for $t$ iterations in \texttt{LocEnc}. According to the sequential composition property of DP, we can conclude that \textsf{VertiMRF} implemented as in Algorithm~\ref{alg:overall_workflow} satisfies $(\epsilon, \delta)$-DP.

% obtain the privacy guarantee.

% \begin{theorem}[\textbf{Privacy Analysis}]\label{the:privacy_analysis_overall}
% Algorithm~\ref{alg:overall_workflow}, with \texttt{LocMRF} implemented as described in Section~\ref{subsec:reuse_privmrf} and \texttt{LocEnc} implemented as in Algorithm~\ref{alg:fm_LocEnc}, satisfies $(\epsilon, \delta)$-DP.
% \end{theorem}

\mypara{Communication cost.}
There is one communication round between each party $\mathcal{P}_i$ and the central server in \textsf{VertiMRF}. The communication includes encoded attributes $\mathcal{M}_i$ and the local MRF information $\{\texttt{MRF}_i, \mathcal{S}_i, \mathcal{G}_i\}$. For sketch-based \texttt{LocEnc}, $\mathcal{M}_i$ contains $t\sum_{A^j \in \mathcal{A}_i} u_j$ sketches. For FO-based \texttt{LocEnc}, $\mathcal{M}_i$ contains a noisy version of local dataset.
$\texttt{MRF}_i$ is parameterized by a vector $\Theta$ with length $\sum_{M\in S_i}\prod_{A^j\in M}u_j$, controlled by the maximal clique size $\tau^{\prime}$ for each local MRF. $\mathcal{G}_i$ is represented by a $(\vert\mathcal{A}_i\vert \times \vert\mathcal{A}_i\vert)$-dimensional adjacent matrix, with $\vert\mathcal{A}_i\vert < d$. The information in $\mathcal{S}_i$, which contains serveral attribute tuples, can be ignored in terms of communication costs.  Considering a total of $m$ parties, the communication cost of \textsf{VertiMRF} is $O(td\bar{u}) + O(d^2) + O(m\tau^{'})$ when using sketch-based \texttt{LocEnc} and $O(nd) + O(d^2) + O(m\tau^{'})$ when using FO-based. Here $\bar{u}$ represents the average domain size of attributes.

\section{MRF generation in central server}\label{sec:server_mrf}

After receiving local \texttt{MRF}s and encoded attributes $\mathcal{M}$ from all parties, the process of the central server can be divided into the following phases: generating the global attribute graph (\textbf{Phase 3}), initializing the marginal set thereby estimating the MRF parameter (\textbf{Phase 4}), refining the marginal set thereby optimizing the MRF parameter (\textbf{Phase 5}) and finally sampling the synthetic data (\textbf{Phase 6}).

\subsection{\textbf{\texttt{GraphCom} in Phase 3}} Since the local attribute graphs are disjoint and each one accurately represents the correlation among a subset of attributes, a basic approach to creating a global attribute graph is to combine the disjoint graphs by linking up certain cross-party attribute pairs. However, there are two constraints (CSTR) that must be satisfied when selecting such cross-party attribute pairs, denoted as $(A^i, A^j)$:

\noindent{$\bullet$CSTR1: $(A^i, A^j)$ should exhibit strong correlation.}

\noindent{$\bullet$CSTR2: The domain size of maximal cliques in the resulting attribute graph should not exceed a predefined threshold value $\tau$.}

To satisfy CSTR1, the central server estimates the  R-score~\cite{cai2021data} $R(A^i, A^j)$ for each cross-party attribute pair $(A^i, A^j)$ with \texttt{CarEst} approach introduced in Section~\ref{sec:mar_est} over the received encoded attributes $\mathcal{M}$:
$
R(A^i, A^j) \approx  \frac{\hat{n}}{2}\left\Vert\frac{\hat{T}_{(A^i,A^j)}}{\hat{n}} - \frac{\hat{T}_{A^i}}{\hat{n}}\cdot \frac{\hat{T}_{A^i}}{\hat{n}}\right\Vert_1 ,
$
where $\hat{T}$ denotes the estimated contingency histogram.
As explained in Section~\ref{subsec:privmrf}, attribute pairs with higher R-scores indicate stronger correlation. After the estimation, the server sorts all attribute pairs in descending order based on their estimated R-scores and greedily connects them in the global attribute graph.

For CSTR2, whenever a link between cross-party attributes is added to $\mathcal{G}$, the server checks the domain size of the maximal clique in the triangulated $\mathcal{G}$ to ensure it does not exceed $\tau$. 

$\tau$ is always set empirically to strike the tradeoff between the model utility and computation complexity. A larger $\tau$ enables more flexible marginal selection but incur high computational efficiency. According to our observation, $[10^5, 5\times10^6]$ is an suitable range for $\tau$. If CSTR2 is satisfied, the process of adding links continues. This process continues until it is no longer possible to satisfy CSTR2. 

% It is important to note that since the server does not have access to the raw data, the centralized estimation of $R(A^i, A^j)$ is not possible. To enable cross-party computation, we utilize the \texttt{CarEst} approach presented in Section~\ref{sec:mar_est} to estimate contingency histograms for $A^i$, $A^j$, and $(A^i, A^j)$ based on the encoded attributes $\mathcal{M}$. Using Equation~(\ref{equ:R-score}), we can then calculate the noisy $R(A^i, A^j)$.

% Once the global attribute graph is generated, the subsequent task is to sample a set of marginals from the graph to learn the parameters $\Theta$ of the global MRF. The complete parameter learning process mainly includes two phases: \texttt{InitMRF} and \texttt{OptMRF}. \texttt{InitMRF} initializes $\Theta$ to enables the MRF can estimate each marginal as more accurate as possible, then \texttt{OptMRF} further refine the $\Theta$ by optimization on the marginals with large estimation error.

\subsection{\texttt{InitMRF} in Phase 4} 

After generating the global attribute graph, the next step is to construct the global MRF. As shown in Section~\ref{subsec:privmrf}, the MRF construction process essentially is to learn a parameter vector $\Theta$ on a marginal set $\mathcal{S}$. In PrivMRF, $\mathcal{S}$ is first initialized by selecting the most highly correlated marginal for each attribute and then refined through adding marginals which cannot be accurately inferred by the global MRF. Meanwhile, $\Theta$ is initialized and optimized by reducing the error of inferring the marginals in $\mathcal{S}$ using a mirror descent algorithm~\cite{cai2021data}. The initialization of $\mathcal{S}$ serves to initialize the global MRF and select an reliable direction for the subsequent refinement of $\mathcal{S}$ and learning of $\Theta$. we follow this process and initialize an $\mathcal{S}$ to initialize the global MRF.

However, unlike PrivMRF, the central server in our setting lacks access to the raw data, it is not practical to compute sufficiently accurate correlations among each attribute with multiple marginals to select the highly correlated ones. More severely, the true value of the contingency histograms are unavailable to compute the inferring error of $\Theta$. Therefore, it is essential to select marginals that can be accurately estimated based on the DP shared information of local parties. With the observation that local \texttt{MRF}s can estimate their marginals with relative low error, we take the union of the local marginal sets as the initialized $\mathcal{S}$, that is $\mathcal{S} = \cup_{i=1}^{m}\mathcal{S}_i$ and take the local \texttt{MRF} inferred contingency histograms $\cup_{i=1}^{m}\{\texttt{MRF}_i(M)\vert \forall M\in \mathcal{S}_i\}$ as the ground truth of contingency histograms, where $\texttt{MRF}_i(M)$ refers to the inferred result of $\texttt{MRF}_i$ on $M$. Based on this ground truth and the initialized $\mathcal{S}$, the server initializes the global MRF.

By observing Equation~(\ref{equ:MRF_distri}), we notice that an MRF encodes the correlation among multiple attributes by representing the marginals in the marginal set. Based on this observation, We can say that the inferred results of each $\texttt{MRF}_i$ on the marginals in the marginal set $\mathcal{S}_i$ encapsulate all the "knowledge" encoded by the MRF. Therefore, such initialization can also be viewed as transferring the knowledge from the local \texttt{MRF}s to the global MRF.
% The marginal selection process should serve two purposes: (i) the inner-party marginals can be inferred by the corresponding local \texttt{MRF}s, and (ii) the cross-party marginals should introduce as little noise as possible. The first purpose is reasonable since the generated attribute graph may produce a clique containing the attributes originally in different cliques in one party. In such a case, the local MRF cannot infer the corresponding contingency histogram. The second purpose is direct.

% To serve the first purpose, we generate an initial marginal set by taking the union of the received local marginal sets. Since we did not remove the existing links in the local attribute graphs when generating the global graph, then the marginal in the local marginal sets must be still contained in one of the cliques in the local attribute graphs, thus can be inferred by the local \texttt{MRF}s. To serve the second purpose, we further generate a candidate marginal set $\mathcal{S}^{c}$ consisting of low-way cross-party marginals. In particular, to limit the introducing noise, we control the domain size of the cross-party marginal $M$ to be smaller than a threshold $d^{c}$, that is $$\prod_{A^i\in M}\vert u_i\vert\leq d^{c}.$$ The reason behind this is that if the domain size is too large, the data number in each cell of the contingency histogram would be too small to be accurately estimated.

\subsection{\texttt{OptMRF} in Phase 5} \label{sec:global_mrf}
After initialization, the encoded correlation in the local \texttt{MRF}s has already been transferred to the global MRF. However, the correlation among cross-party attributes has not been characterized by the global MRF. To address this, the central server further refines $\mathcal{S}$ by inserting cross-party marginals whose contingency histograms can be estimated using the \texttt{CarEst} approach over the encoded local attributes $\mathcal{M}$. To minimize noise, we mainly select the low-way cross-party marginals denoted as $\mathcal{S}^c$. Specifically, the average count in each cell of $T_M$ of each marginal $M \in \mathcal{S}^c$ is controlled to be larger than a threshold $d^{c}$, as given by $\frac{\hat{n}}{\prod_{A^i\in M}\vert u_i\vert}\geq d^{c}$, where $d^{c}$ controls the error of the estimation of \texttt{CarEst}, which is also set empirically. As shown in Theorem~\ref{the:fm_error}, accurate estimation of CarEst becomes challenging with small average counts in the contingency histogram.
It is worth noting that the optimization of $\Theta$ involves multiple rounds. In each round, the server randomly samples several cross-party marginals from $\mathcal{S}^c$ and optimize $\Theta$ mainly on the ones with significant inferring error, as measured by the L1 distance between the inferred histograms of the global MRF and the true values estimated using \texttt{CarEst} over the encoded attributes $\mathcal{M}$.

Once the global MRF is constructed, approximating the data distribution and sampling synthetic data becomes a straightforward task. For detailed information, please refer to ~\cite{cai2021data}.

\section{Dimension Reduction and Consistency}\label{sec:impro}
While \textbf{Phase 1 - 5} with details introduced in Section 4 and Section 5 can compose a complete algorithm for differentially private vertical data synthesis, we encounter a dilemma when optimizing the algorithm by tuning the granularities of the attributes.
With a coarser granularity, the \texttt{LocEnc-CarEst} can have smaller relative errors, but the \texttt{LocMRF} and global MRF becomes inferior to its best performance with more fine-grained granularity.
Thus, configuring different granularities for those parts can be an alternative improvement.
However, there are two issues for this inconsistent granularity solution: 1) how to reduce dimension while keeping as much information as possible; 2) how to enforce consistency between the frequencies of different granularities.

% This section introduces two techniques aimed at improving the estimation of cross-party marginals, particularly when the attributes have large domain sizes.

\subsection{Dimension reduction}\label{subsec:dim_rec}

As stated in Theorem~\ref{the:fm_error}, when the domain sizes of attributes increase, the estimated cross-party marginals by \texttt{CarEst} can deviate significantly. This deviation occurs because the expected number of data points in each cell of the contingency histogram decreases. To address this issue, we propose binning the attributes to reduce the domain sizes. The binned attributes are encoded using \texttt{LocEnc} and sent to the server (\textbf{Phase 2}). The encoded attributes are then used to calculate R-scores in the \texttt{GraphCom} procedure (\textbf{Phase 3}) and estimate the cross-party marginals to optimize the global MRF (\textbf{Phase 5}). However, since the global MRF is constructed based on the raw attributes without binning, the estimated contingency histogram cannot be used directly. To overcome this, we employ a histogram recovery technique to transform the estimated low-dimensional histograms of the binned attributes into the high-dimensional ones.

\textbf{\textit{The basic idea is to approximate the high-dimensional distributions using low-dimensional ones.}} According to the joint distribution formula, when considering $(A,B)$ as a marginal for estimation, where $X$ and $Y$ are the binned versions of $A$ and $B$ respectively, the high-dimensional marginal distribution can be estimated as $Pr[A,B] \approx \sum_{(X,Y)}Pr[X, Y]\cdot Pr[A|X]\cdot Pr[B|Y].$
Here, $Pr[X, Y]$ represents the low-dimensional distribution over the binned attributes, while $Pr[A|X]$ and $Pr[B|Y]$ are referred to as value distributions and are also low-dimensional. Figure~\ref{fig:histogram} provides a visualization of our dimension reduction technique. 
% Further details are explained below.

\begin{figure}[th]
    \vspace{-0.12in}
	\centering
	\includegraphics[width=0.48\textwidth]{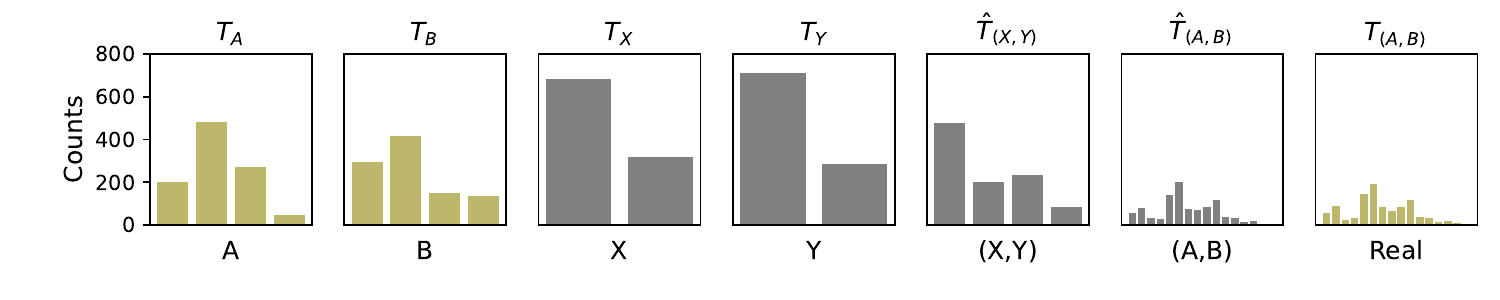}
    \vspace{-0.25in}
    \caption{Instantiation of Dimension Reduction. \label{fig:histogram}}
    \vspace{-0.17in}
\end{figure}

\noindent \textbf{Attribute binning.} Each party $\mathcal{P}_i$ applies equal-width binning to each local attribute $A^j\in\mathcal{A}_i$ before \texttt{LocEnc}. The number of bins $b$ is specified by the server or negotiated by data parties. Equal-width binning is based solely on the domain size of each attribute and does not expose any statistical information from the raw data. However, to facilitate the subsequent histogram recovery, it is necessary to preserve the value distribution within each bin for each attribute during the binning process and send it to the server.

For an attribute $A^j$ with a domain size of $u_j$, the values are allocated to $b$ bins with equal width. Suppose the values ${v_i^1,..., v_i^k}$ are allocated to the $l$-th bin, and their corresponding frequencies are ${n_i^1,..., n_i^k}$. The distribution of the $l$-th bin of $A^j$ is defined as: $$U_{(j,l)} = \left[U_{(j,l)}^1,\ldots,U_{(j,l)}^k\right] \triangleq \left[\frac{n_i^1}{\sum_{h=1}^{k}n_i^h},\ldots, \frac{n_i^k}{\sum_{h=1}^{k}n_i^h}\right].$$ Since the frequencies are obtained from the raw data, the resulting value distributions may expose sensitive statistical information. To ensure privacy, we utilize the Laplacian mechanism to perturb the frequencies with a sensitivity of $1$ and a privacy budget of $\epsilon'$. The value distributions are then computed based on the noisy frequencies. Considering the sequential composition of DP, the total privacy cost of the overall binning procedure should be $d\epsilon'$. 
% Further details can be found in Algorithm~\ref{alg:dimension_red} (lines 1-12).

\noindent \textbf{Histogram recovery (\texttt{HisRec}).} Let's assume the contingency histogram of the marginal $(A^i,A^j)$ estimated by \texttt{CarEst} over the binned attributes is denoted as $\hat{T}^{(low)}$. In $\hat{T}^{(low)}$, $\hat{T}^{(low)}[(v^{\prime}_{(i)},v^{\prime}_{(j)})]$ represents the number of data points falling in the $v^{\prime}_{(i)}$-th bin of $A^i$ and the $v^{\prime}_{(j)}$-th bin of $A^j$. We can recover the high-dimensional histogram $\hat{T}^{(high)}$ by estimating the number of data points falling in each cell where $A^i = v_{(i)}$ and $A^j = v_{(j)}$ as:
$$\hat{T}^{(high)}\left[(v_{(i)},v_{(j)})\right] = \hat{T}^{(low)}\left[(v^{\prime}_{(i)}, v^{\prime}_{(j)})\right]\cdot U_{(i,v^{\prime}_{(i)})}^{h^{\prime}_{(i)}}\cdot U_{(j,v^{\prime}_{(j)})}^{h^{\prime}_{(j)}}.$$
where $h^{\prime}{(i)}$ and $h^{\prime}{(j)}$ represent the value index of $v_{(i)}$ allocated to the $v^{\prime}_{(i)}$-th bin of $A^i$ and $v_{(j)}$ allocated to the $v^{\prime}_{(j)}$-th bin of $A^j$. 
% By iterating through the cells in $\hat{T}^{(low)}$ and performing the above estimation, $\hat{T}^{(high)}$ can be completely recovered.
Although demonstrated with the two-way marginal case, this technique can be easily extended to higher-way marginal cases.
% Algorithm~\ref{alg:dimension_red} provides further details (lines 13-21).

\subsection{Consistency enforcement}\label{subsec:consis}

% As discussed in Section~\ref{sec:global_mrf}, estimating the contingency histograms for all intra-party and cross-party marginals in the marginal set $\mathcal{S}$ is necessary to construct the global MRF. The intra-party marginals are estimated using local \texttt{MRF}s, while cross-party marginals are estimated using \texttt{CarEst}. Nevertheless, variations in the sources of randomness can cause inconsistencies between the estimated histograms from local \texttt{MRF}s and \texttt{CarEst} for specific attribute sets, resulting in excessive noise.
As discussed in Section~\ref{sec:global_mrf}, estimating contingency histograms for intra-party and cross-party marginals in the marginal set $\mathcal{S}$ is vital for constructing the global MRF. Intra-party marginals are estimated using local \texttt{MRF}s, while cross-party marginals are estimated using \texttt{CarEst}. Nevertheless, variations in the sources of randomness can cause inconsistencies between the estimated histograms from local \texttt{MRF}s and \texttt{CarEst} for specific attribute sets.

Let's consider a two-way cross-party marginal, denoted as $(A^i, A^j)$. The contingency histogram estimated by \texttt{CarEst} is denoted as $\hat{T}_{(A^i, A^j)}$, or simply $\hat{T}$. If we marginalize $\hat{T}$ to obtain $\hat{T}_{(A^i)}$ and $\hat{T}_{(A^j)}$, these results may exhibit inconsistencies with the histograms $\widetilde{T}_{(A^i)}$ and $\widetilde{T}_{(A^j)}$ inferred from local \texttt{MRF}s.

To address this inconsistency, we employ a two-step technique to ensure consistency among ${\hat{T}, \widetilde{T}_{(A^i)}, \widetilde{T}_{(A^j)}}$. Firstly, we transform all three contingency histograms into marginal distributions by normalizing them. For simplicity, we continue to use the notation $T$ to represent the marginal distribution. The two-step technique operates on each attribute individually, taking $A^i$ as an example.

% \textbf{\textit{The first step is the \underline{consistency} step}}.
\mypara{Step 1: consistency.}
We begin by establishing agreement between $\widetilde{T}_{(A^i)}$ and $\hat{T}_{(A^i)}$ by taking their arithmetic mean: $\bar{T} = \frac{\widetilde{T}_{(A^i)}+\hat{T}_{(A^i)}}{2}$. We then update both $\hat{T}$ and $\widetilde{T}_{(A^i)}$ to be consistent with $\bar{T}$. Specifically, $\widetilde{T}_{(A^i)}$ is directly set to $\bar{T}$. As for $\hat{T}$, changing a cell in $\hat{T}_{(A^i)}$ would affect $u_j$ cells in $\hat{T}$ (where $u_j$ is a constant). To maintain the unchanged marginal distribution $\hat{T}_{(A^j)}$, we calculate the difference between $\hat{T}{(A^i)}$ and $\bar{T}$ for each cell when $A^i$ takes a specific value $v_{(i)}$, and then distribute the difference equally among all $u_j$ affected cells:  $\hat{T}[(v_{(i)},v_{(j)})] = \hat{T}[(v_{(i)},v_{(j)})] + \frac{\bar{T}[v_{(i)}]- \hat{T}_{(A^i)}[v_{(i)}]}{u_{j}}.$
% where $\hat{T}_{(A=v_{a}, B)}$ denotes the slice of $\hat{T}_{(A, B)}$ with $A$ taking value $v_{a}$.

% \textbf{\textit{The second step is the \underline{normalization} step}}. 
\mypara{Step 2: normalization.}
After the consistency step, negative numbers may appear in the marginal distribution. To ensure non-negativity, we set all negative numbers to 0. However, this adjustment may cause the sum of the distribution $\hat{T}$ to exceed 1. To address this, we re-normalize $\hat{T}$. Nevertheless, this re-normalization may introduce inconsistency between $\hat{T}_{(A^i, A^j)}$ and $\widetilde{T}{(A^i)}$ again. To mitigate this issue, we repeat the consistency and normalization process for each attribute multiple times until the resulting inconsistency becomes negligible. Finally, the consistent marginal distribution is transformed into a contingency histogram by multiplying it with a DP-sanitized data number $\hat{n}$. 

This consistency enforcement technique is employed in \textbf{Phase 5} and can be directly extended to higher-way marginal cases. 

% The above technique can be directly extended to higher-way marginal cases. The consistency enforcement technique is employed in \textbf{Phase 5}.
% For additional details, please refer to Algorithm~\ref{alg:consistency}. 

\section{Experiments}\label{sec:experiments}
In this section, we first conduct~\footnote{The code is available at \url{https://github.com/private-mechanism/Verti_MRF}} the end-to-end comparisons of \textsf{VertiMRF} to baseline methods. Then, we validate each component by conducting the ablation experiments (shown in Appendix~\ref{sec:add_exp}). The final experimental results demonstrate the superiority of \textsf{VertiMRF}.

\subsection{Experiment settings}

\noindent\textbf{Datasets.} We evaluate our algorithms on four datasets in Table~\ref{tab:data}.
\begin{itemize}[itemindent=2ex, listparindent=0ex, leftmargin=0ex, labelwidth=0px]
\item\textbf{Adult~\cite{asuncion2007uci}.} The data is sourced from the United States Census Bureau. It consists of 45,222 instances, each containing 15 attributes capturing demographic and socio-economic information, such as age, race, education, and income level.

\item\textbf{NLTCS~\cite{manton2010national}.} The data is collected from a study on health status of older adults. It includes 21,574 records and 16 attributes that describe demographic information and health conditions.

\item\textbf{BR2000~\cite{steven2015ipums}.} The data originates from a population Census in Brazil and contains 38,000 records. It includes 13 attributes that provide information about demographic, and economic aspects.

\item\textbf{Fire~\cite{ridgeway2021challenge}.} The data includes records of fire unit responses to calls in San Francisco. It consists of 305,119 records, with each record containing 15 attributes.
% including the call number, address, and alarm information.
\end{itemize}

\begin{table}[]
\centering
\setlength\tabcolsep{3pt}
\caption{Characteristics of Datasets}
% \vspace{-0.3cm}
\label{tab:data}
\begin{tabular}{|c|ccccc|}
\toprule
 \textbf{Dataset}&  \textbf{Records}  & \textbf{Attrs}  &\textbf{Dom. } & \textbf{Dom. Size}& \textbf{Attr. Split} \\ \hline
 NLTCS&  21574&  16&  2&  $\approx 6\times 10^4$& 8\&8\\ 
 Adult&  45222&  15&  2-42&  $\approx 4\times 10^{14}$& 8\&7\\ 
 BR2000&  38000&  13&  2-21&  $\approx 3\times 10^{9}$& 7\&6\\ 
 Fire&  305119&  15&  2-46&  $\approx 1\times 10^{15}$& 8\&7\\
\bottomrule
\end{tabular}
\end{table}

\begin{figure*}[t]
	\centering
	\includegraphics[width=1.0\textwidth]{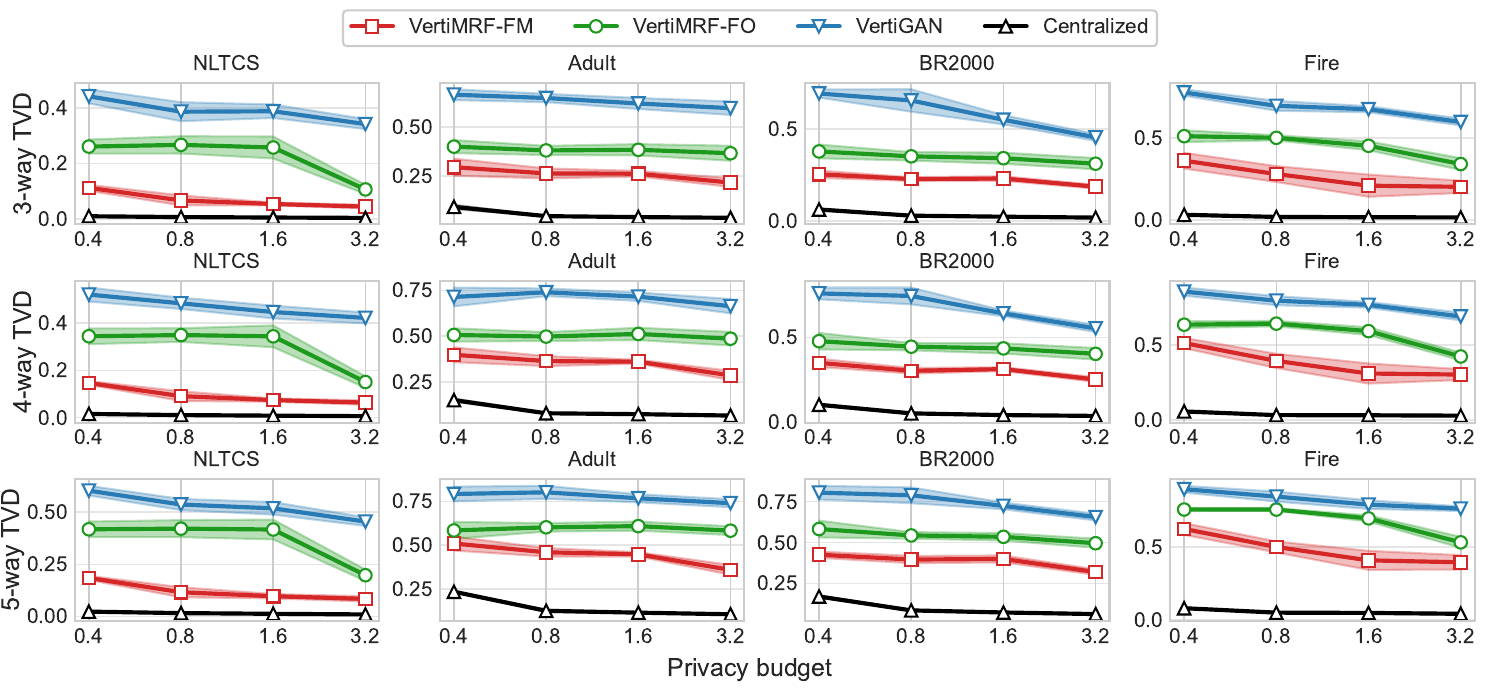}
    \caption{$l$-way TVD vs. privacy budget $\epsilon$. \label{fig:tvd}}
\end{figure*}

\noindent\textbf{Metrics.} We evaluate the performance based on two metrics. 
\begin{itemize}[itemindent=2ex, listparindent=0ex, leftmargin=0ex, labelwidth=0px]
\item \textbf{$l$-way TVD.} We randomly sample 300 marginals with $l$ attributes each from the synthetic data. For each marginal, we compute the total variation distance (TVD) between the raw and synthetic data. We calculate the average TVD for $l$ values of {3, 4, 5} across all marginals and report the average measurement based on 10 iterations.

\item \textbf{Misclassification rate.} 
% We utilize synthetic data to train SVM classifiers for predicting specific attributes based on all other attributes. The predicted attribute for each dataset is shown in the "label" row of Table~\ref{tab:data}. For NLTCS, each attribute is taken as the label to predict, and the final result reports the average. We use $80\%$ of the raw data to generate synthetic data and train the classifier, while the remaining data is used as the testing set to report the misclassification rate. We employ 5-fold cross-validation and report the average.
We employ synthetic data to train SVM classifiers for predicting specific attributes based on all other attributes. The predicted attribute for each dataset is shown in the "label" row of Table~\ref{tab:data}. For NLTCS, each attribute serves as the label to predict, and the average result is reported. We use $80\%$ of the raw data to generate synthetic data and train the classifier, while the remaining data is used as the test set to report the misclassification rate. We utilize 5-fold cross-validation and report the average.
\end{itemize}

\noindent\textbf{Compared methods.} We compare the following methods.
\begin{itemize}[itemindent=2ex, listparindent=0ex, leftmargin=0ex, labelwidth=0px]

\item {\textbf{\textsf{VertiGAN}}} 
~\cite{jiang2023distributed} employs a partitioned GAN with a multi-output global generator and multiple local discriminators. To ensure privacy, local discriminators are trained with DP-SGD~\cite{abadi2016deep} using raw data. The global generator is updated by aggregating the local gradients. The privacy budget is fully utilized during the DP-SGD procedure of local discriminator.

\item{\textbf{\textsf{Centralized}}} refers to PrivMRF~\cite{cai2021data} in the centralized setting.

\item{\textbf{\textsf{VertiMRF-FO} $\&$ \textsf{VertiMRF-FM}.}} Both methods are based on our proposed \textsf{VertiMRF} framework, with one equipped with FO-based and the other with sketch-based \texttt{LocEnc} approach.
\end{itemize}

In addition, we encountered the \textsf{DPLT} approach~\cite{tang2019differentially}, which utilizes a latent tree structure to capture attribute correlations. Despite our diligent efforts to replicate \textsf{DPLT}, we encountered significant computation overhead when calculating the marginal distribution of high-level latent attributes. We also faced ambiguity regarding data synthesis from the constructed tree. As a result, we have chosen not to include \textsf{DPLT} in our comparison.

\mypara{Parameter setting.} In our experiments, we use default values for \textsf{VertiMRF-FM}, setting the repetition number of the DP FM sketch to $t = 2000$ and $\delta = 1/n$. The network structure of \textsf{VertiGAN} follows the configuration described in the original paper~\cite{jiang2023distributed} and the privacy is tracked with RDP~\cite{mironov2017renyi}. For all datasets except NLTCS, we set the binning number to $b=4$. As for the privacy budget allocation, we allocate $40\%$ to \texttt{LocMRF}, $40\%$ to \texttt{LocEnc} (with $10\%$ of the $40\%$ used to generate a noisy data count $\hat{n}$), and the rest $20\%$ for sanitizing value distributions in the binning procedure. By default, our algorithms are validated in a two-party setting, the attribute distribution on the two parties is shown in the "Attr. Split" row of Table~\ref{tab:data}, e.g., "8\&7" means that $8$ attributes are assigned to one party and other $7$ attributes are assigned to the other one.

% Innovatively, \textsf{DPLT} combines the inner product protocol and a distributed Laplace mechanism to protect the sensitive statistics when calculating the cross-party correlations. 

% and perturbs the tree parameters via a distributed Laplace mechanism to achieve DP guarantee for each local dataset. 

% \noindent\textbf{Simulation Methodology.} For simulating the vertical setting, we randomly partitioned the multiple attributes of each dataset into $M$ parties. Besides, all our experiments are run on XXX.

\subsection{End to End Comparisons}
\noindent\underline{\textbf{Comparison on $l$-way TVD.}}
Figure~\ref{fig:tvd} compares the methods based on the average TVD of the 
$l$-way marginals. As shown, \textsf{VertiMRF-FM} and \textsf{VertiMRF-FO} consistently produce smaller TVD than \textsf{Verti-GAN} regardless of privacy cost or dataset. This demonstrates the superiority of \textsf{VertiMRF}. 
Additionally, \textsf{VertiMRF-FM} outperforms \textsf{VertiMRF-FO} across all cases, indicating that the sketch-based \texttt{LocEnc} and \texttt{CarEst} can provide more accurate estimation of the cross-party marginals compared to the FO-based approaches. It is worth noting that \textsf{VertiGAN} consistently yields significantly larger TVD results. This can be attributed to the fact that GAN-based data synthesis methods are not well-suited for synthesizing tabular data, as discussed in previous studies~\cite{cai2021data, mckenna2019graphical}. 
Furthermore, the advantages of \textsf{Centralized} over \textsf{VertiMRF-FM} become more prominent in datasets with larger domain sizes, such as adult, BR2000, and Fire. Although \textsf{VertiMRF-FM} performs closely to \textsf{Centralized} on NLTCS when $\epsilon$ is larger, the difference becomes more pronounced in the other three datasets due to their larger domain sizes. This aligns with our analysis in Theorem~\ref{the:fm_error}. A larger domain size leads to smaller average count in a contingency histogram, thereby deriving a larger estimation error of \texttt{CarEst}.

\noindent\underline{\textbf{Comparison on SVM classification.}}
Figure~\ref{fig:svm_accuracy} presents the average misclassification rates of the SVM classifiers trained on the synthetic data. \textsf{VertiMRF-FM} consistently outperforms other vertical methods. Misclassification rates of \textsf{VertiGAN} are as high as 
$40\%$ even when $\epsilon = 3.2$, which is significantly larger than both \textsf{VertiMRF-FM} and \textsf{VertiMRF-FO} methods. Additionally, the advantages of \textsf{centralized} over \textsf{VertiMRF-FM} are magnified as the domain size of the dataset increases. These findings align with results shown in Figure~\ref{fig:tvd}, illustrating the effectiveness of \textsf{VertiMRF-FM} and the negative impact of large domain sizes. Similar Results on Br2000 and Fire datasets are shown in Figure~\ref{fig:br2000_svm_accuracy} in Appendix~\ref{sec:add_exp}.
\begin{figure*}[t]
	\centering
	\includegraphics[width=1.0\textwidth]{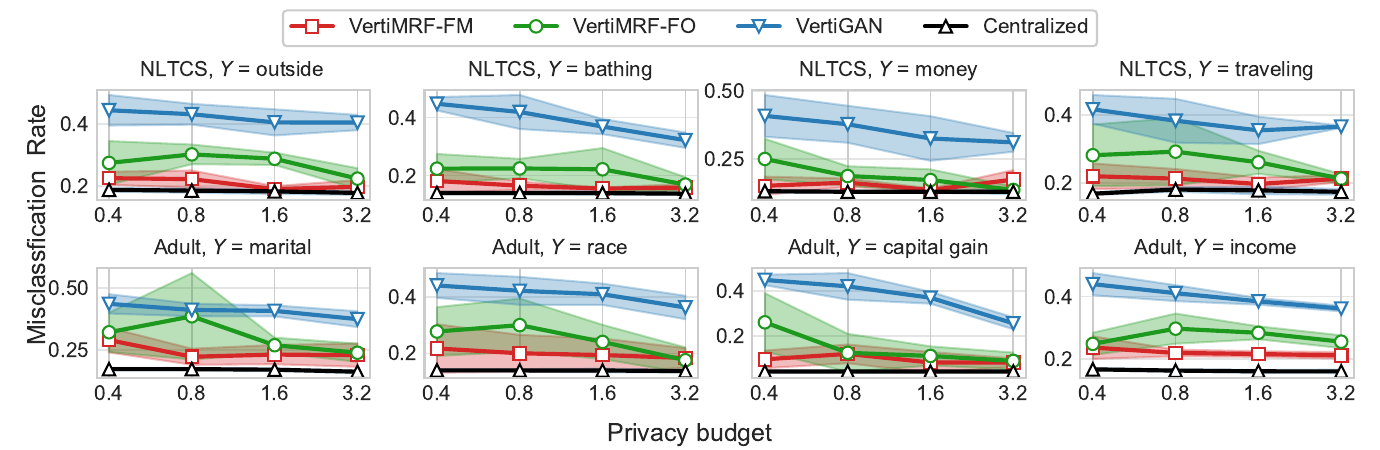}
    \caption{SVM misclassification rate vs. privacy budget $\epsilon$. \label{fig:svm_accuracy}}
\end{figure*}

\noindent\underline{\textbf{Impact of the number of parties.}} We examine the impact of the party number $m$ on the utility of synthetic data. Figure~\ref{fig:M} compares the TVD results obtained under different
$m$ settings on the NLTCS and Adult datasets with a privacy budget $\epsilon = 0.8$. In the experiments, $m = ALL$ indicates that the attributes are distributed to 
\textit{d} parties, with each party having one distinct attribute. The results demonstrate that as $m$ increases, the TVD results also increase. This is primarily because when attributes are partitioned to more parties, \texttt{LocMRF} with high precision contributes less to the global MRF.
%is unable to effectively capture the correlations among local attributes. 
In such cases, the \texttt{LocEnc} and \texttt{CarEst} procedures dominate the errors.
\begin{figure}[htbp]
	\centering
	\includegraphics[width=0.45\textwidth]{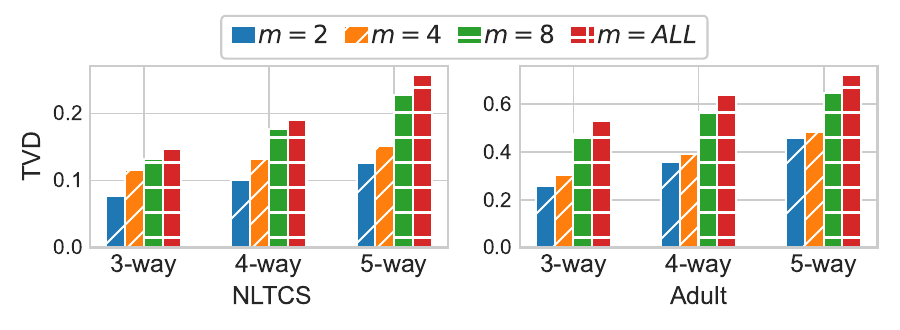}
    % \vspace{-0.2in}
    \caption{Impact of party number ($\epsilon=0.8$). \label{fig:M}}
    % \vspace{-0.15in}
\end{figure}

\begin{figure}[htbp]
	\centering
	\includegraphics[width=0.45\textwidth]{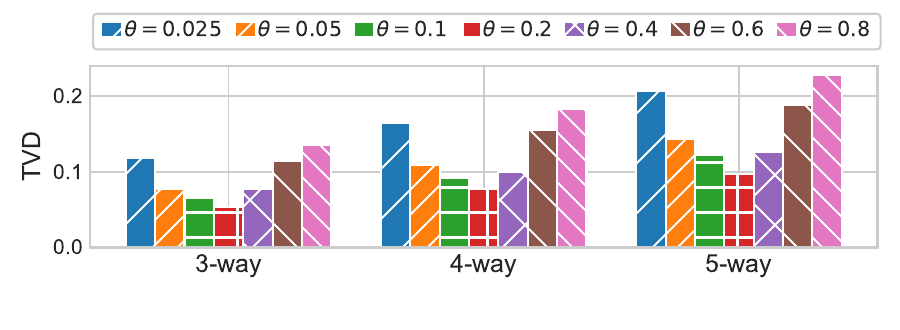}
    % \vspace{-0.25in}
    \caption{Impact of privacy budget allocation ($\epsilon=0.8$). \label{fig:theta}}
    % \vspace{-0.1in}
\end{figure}

\noindent\underline{\textbf{Impact of privacy budget allocation.}} We examine the impact of privacy budget allocation on the utility of synthetic data. In this study, we fix the total privacy budget as $\epsilon = 0.8$ and compare the TVD results on the NLTCS dataset under different privacy budget ratio $\theta$ assigned to \texttt{LocMRF}. The remaining $(1-\theta)$ proportion of the privacy budget is then fully allocated to \texttt{LocEnc}. Figure~\ref{fig:theta} illustrates that as 
$\theta$ increases, the TVD results initially decrease when $\theta \leq 0.2$, but then increase when $\theta > 0.2$. These results highlight the tradeoff between the impacts of \texttt{LocEnc} and \texttt{LocMRF}. \texttt{LocMRF} aids in estimating the intra-party marginals, while \texttt{LocEnc} aids in estimating the cross-party marginals. When 
$\theta$ is small, the error is dominated by inaccurate estimation of the intra-party marginals. Conversely, when 
$\theta$ is larger, the error is mainly caused by inaccurate estimation of the cross-party marginals.

\noindent\underline{\textbf{Communication and computation cost.}} 
In Table~\ref{tab:commu_cost}, we compare the communication costs and computation time of the four methods on Adult dataset. As analyzed in Section~\ref{subsec:analysis}, the communication overhead of \textsf{VertiMRF-FM} is expected to be smaller than that of \textsf{VertiMRF-FO} when $t\bar{u}<n$. Consistent with our analysis, we observe that the overhead of \textsf{VertiMRF-FM} is smaller than that of \textsf{VertiMRF-FO} when $t\bar{u}<n$  with  $t = 2000$ but larger when $t\bar{u}>n$ with $t = 8000$. The communication in \textsf{VertiGAN} involves sending gradients of local generators to the server and broadcasting the updated model to the local parties. Therefore, the overall communication cost depends on the model size and the number of communication rounds.

In terms of computation time, when using the sketch-based \texttt{LocEnc}, each local party needs to perform $tn$ hash mappings, whereas the FO-based \texttt{LocEnc} only requires $n|A_i|$ perturbations. Since $t>>|A_i|$, the FO-based \texttt{LocEnc} requires less computation time. The hash mappings can be accelerated by parallel computation since they run independently. By introducing $40$ parallel threads, the computation time can be significantly reduced. On the server side, the computation time is nearly identical for both \textsf{VertiMRF-FM} and \textsf{VertiMRF-FO}. That's because apart from \texttt{CarEst}, both methods execute identical computations on the server side. Whether it is FO-based \texttt{CarEst} or sketch-based \texttt{CarEst}, the computation process solely involves simple calculations and does not significantly affect the computation time. In \textsf{VertiGAN}, each party generates fake data and computes model gradients, while the server aggregates and broadcasts the updated model. Therefore, the most time consumption occurs at the local party.

\begin{table}[]\small
\setlength\tabcolsep{3pt}
\caption{Communication cost and computation time}
% \vspace{-0.1in}
\label{tab:commu_cost}
% \resizebox{0.80\columnwidth}{!}{%
\begin{tabular}{c|c|c|c|c}
\toprule
\multirow{2}{*}{Dataset} & \multirow{2}{*}{methods} & \multirow{2}{*}{commu. cost} & \multicolumn{2}{c}{compu. time} \\ 
                  &                  &                   &           per party&   server       \\ \hline
% \multirow{5}{*}{NLTCS} &             \textsf{VertiMRF-FM} ($t = 2000$)&                   $4.9$ Mb&         $11.5$ min  &        $4$ min \\
%                   &                   \makecell{\textsf{VertiMRF-FM} \\($t = 2000$, $threads = 40$)}&                   $4.9$ Mb&         $2.1$ min  &        $4$ min \\
%                   &                   \textsf{VertiMRF-FM} ($t = 8000$)&                   $10.7$ Mb&         $47$ min  &        $4.3$ min \\
%                   &                   \textsf{VertiMRF-FO}&                   $5.5$ Mb&         $1.2$ min   &         $4.0$ min  \\
%                   &                   \textsf{VertiGAN}&                   $96$ Mb&          $8.1$ min  &   $10$ s       \\
%                   &                   Centralized &                     - &          - &   $84$ s        \\ \hline
\multirow{5}{*}{Adult} &              \textsf{VertiMRF-FM} ($t = 2000$)&                   $15$ Mb&        $23.1$ min &       $67$ min\\
                  &                   \makecell{\textsf{VertiMRF-FM} \\($t = 2000$, $threads = 40$)}&                   $4.9$ Mb&         $4.1$ min  &        $67$ min \\
                  &                   \textsf{VertiMRF-FM} ($t = 8000$)&                   $22$ Mb&         $93$ min  &        $67$ min \\
                  &                   \textsf{VertiMRF-FO}&                   $18$ Mb&          $2.5$ min &       $67$ min \\
                  &                   \textsf{VertiGAN}&                   $112$ Mb&         $8.3$ min  &  $10$ s        \\
                  &                   Centralized&                   - &          - & $16$ min   \\ 
\bottomrule
\end{tabular}
% }
% \vspace{-1 in}
\end{table}

\begin{table}[]\small
\setlength\tabcolsep{3pt}
\caption{3-way TVD under different attribute distributions}
\label{tab:splits}
\begin{tabular}{c|c|c|c|c}
\toprule
Splitters & Params. & \textsf{VertiMRF-FM} & \textsf{VertiMRF-FO} &\textsf{VertiGAN} \\ \hline
\multirow{5}{*}{Importance} &         0.1&                   \textbf{0.0583 ($\pm$0.005)}&         0.234($\pm$0.023)	&0.426 ($\pm$0.027) \\
                  &                   1&                 \textbf{0.0667 ($\pm$0.021)}&	0.249 ($\pm$0.017)&	0.430 ($\pm$0.056) \\
                  &                   10&                   \textbf{0.0589 ($\pm$0.006)}&	0.257 ($\pm$0.019)&	0.458 ($\pm$0.081) \\
                  &                   100&                   \textbf{0.0648 ($\pm$0.007)}& 	0.266 ($\pm$0.022)& 	0.465 ($\pm$0.068)\\ 
                  \hline
\multirow{5}{*}{Correlation} &         0&                   \textbf{0.0735 ($\pm$0.007)}&	0.261($\pm$0.027)&	0.436 ($\pm$0.034) \\
                  &                   0.3&                 \textbf{0.0524 ($\pm$0.006)}&	0.296 ($\pm$0.024)&	0.416 ($\pm$0.031) \\
                  &                  0.6& \textbf{0.0684 ($\pm$0.006)}&	0.272 ($\pm$0.023)&	0.438 ($\pm$0.038) \\
                  &                 1.0&  \textbf{0.0678 ($\pm$0.009)}&	0.281 ($\pm$0.034)&	0.401 ($\pm$0.042)\\ 
\bottomrule
\end{tabular}
\end{table}

\noindent\underline{\textbf{Impact of different attribute distributions.}} We calibrate the importance and correlation of attributes from different data parties based on the attribute splitters proposed for VFL tasks in Vertibench~\cite{wu2023vertibench}, thereby evaluating the impact of varying attribute distributions on algorithm performance.
Table~\ref{tab:splits} summarizes the resulting 3-way TVD results, i.e., mean and standard deviation across 5 independent runs, under different parameter settings for each algorithm on NLTCS dataset. As shown, the superiority of \textsf{VertiMRF-FM} on other baseline algorithms is significant and stable with respect to different splitting strategies. Furthermore, the TVD results for all algorithms fluctuate within a narrow range as parameters $\alpha$ and $\beta$ vary, indicating that the performance of these algorithms is robust against variations in feature splits. 

\section{Related Work}\label{sec:related_work}
We review related work from the following three perspectives. More detailed related work can be referred to \cite{zhao2023fullversion}.

\noindent\textbf{DP data synthesis.}
There have been plenty of approaches~\cite{aydore2021differentially,li2014dpsynthesizer,ge2020kamino,li2021dpsyn,mckenna2019graphical,ren2018lopub,li2014differentially} to generate synthetic data with DP guarantee, which can be categorized into GAN-based ~\cite{beaulieu2019privacy,abay2019privacy,zhang2018differentially,frigerio2019differentially,chen2020gs,jordon2018pate}, game-based ~\cite{gaboardi2014dual,vietri2020new,hardt2012simple}, and marginal-based approaches~\cite{mckenna2019graphical,mckenna2021winning,zhang2017privbayes,cai2021data,zhang2021privsyn}. Among them, the marginal-based ones tend to perform best, aiming to approximate the joint distribution of high-dimensional data with multiple low-way marginals. Such an approximation can help to circumvent the curse of dimensionality, i.e., the exponentially exploded sizes of the contingency histogram with the increased attribute number. For example, PrivBayes~\cite{zhang2017privbayes} utilizes the Bayesian network to select low-way marginals to approximate a high-dimensional distribution. PrivMRF ~\cite{cai2021data} applies a Markov Random Field to model the data distribution, which enables flexible selection of low-way marginals. Without learning a graph structure, PrivSyn~\cite{zhang2021privsyn} greedily searches numerous low-way marginals to represent and synthesize the original dataset directly. Despite high utility with DP guarantee, these approaches cannot be directly extended to the vertical federated setting.

\noindent\textbf{Private vertical data synthesis.}
There are several works~\cite {mohammed2011anonymity,jiang2006secure,mohammed2013secure,tang2019differentially,jiang2023distributed} on the private data synthesis under vertical setting. Among those works, some are based on the privacy model of k-anonymity ~\cite{sweeney2002k}, which has been proven to be vulnerable to various privacy attacks~\cite{kifer2009attacks,wong2007minimality}. A few works~\cite{mohammed2013secure,tang2019differentially,jiang2023distributed} explore DP data synthesis under a vertical setting. For instance, \cite{mohammed2013secure} proposes a two-party DP data synthesis framework relying on a given taxonomy tree, which is designed for classification tasks. \cite{tang2019differentially} utilizes a latent tree model to capture the correlations among cross-party attributions. Besides, DP-WGAN is also adapted to the vertical setting~\cite{jiang2023distributed} to generate synthetic data. To the best of our knowledge, we are the first work adapt the marginal-based approach to the vertical setting. The empirical results have demonstrated the superiority.

\noindent\textbf{Vertical data analysis with DP.} Apart from data synthesis, there are also several works on the DP computing~\cite{he2022differentially,groce2019cheaper} and DP machine learning~\cite{li2023differentially,chen2020vafl,wu2020privacy,xie2022improving} under vertical setting. In particular, the work~\cite{groce2019cheaper} applies DP to protect the loads of hash table for achieving malicious-secure two-party private set intersection. Another work~\cite{xie2022improving} enables each data party to build a local feature extractor to output DP-sanitized feature embedding for realizing vertical deep learning with DP. A recent paper~\cite{li2023differentially} achieves DP vertical k-means by leveraging the inherent randomness of FM-sketch to protect the membership information of data points.

\section{conclusion}\label{sec:conclusion}
We have presented \textsf{VertiMRF}, a novel differentially private algorithm to generate synthetic data in the vertical federated setting. In particular, we applied DP FM-sketch to encode the local data of each party and estimate cross-party marginals. Based on the shared sketches and local MRFs constructed by local parties, the central server can build an MRF to represent global correlations without access to the raw data and violation of DP. Additionally, we also provided two techniques tailored for datasets with large attribute domain sizes. Finally, we empirically validated \textsf{VertiMRF} by conducting end-to-end comparisons and ablation studies.

{
\bibliographystyle{abbrv}
\bibliography{ref}

\begin{thebibliography}{10}

\bibitem{abadi2016deep}
M.~Abadi, A.~Chu, I.~Goodfellow, H.~B. McMahan, I.~Mironov, K.~Talwar, and
  L.~Zhang.
\newblock Deep learning with differential privacy.
\newblock In {\em Proceedings of the 2016 ACM SIGSAC conference on computer and
  communications security}, pages 308--318, 2016.

\bibitem{abay2019privacy}
N.~C. Abay, Y.~Zhou, M.~Kantarcioglu, B.~Thuraisingham, and L.~Sweeney.
\newblock Privacy preserving synthetic data release using deep learning.
\newblock In {\em Machine Learning and Knowledge Discovery in Databases:
  European Conference, ECML PKDD 2018, Dublin, Ireland, September 10--14, 2018,
  Proceedings, Part I 18}, pages 510--526. Springer, 2019.

\bibitem{asuncion2007uci}
A.~Asuncion and D.~Newman.
\newblock Uci machine learning repository, 2007.

\bibitem{aydore2021differentially}
S.~Aydore, W.~Brown, M.~Kearns, K.~Kenthapadi, L.~Melis, A.~Roth, and A.~A.
  Siva.
\newblock Differentially private query release through adaptive projection.
\newblock In {\em International Conference on Machine Learning}, pages
  457--467. PMLR, 2021.

\bibitem{beaulieu2019privacy}
B.~K. Beaulieu-Jones, Z.~S. Wu, C.~Williams, R.~Lee, S.~P. Bhavnani, J.~B.
  Byrd, and C.~S. Greene.
\newblock Privacy-preserving generative deep neural networks support clinical
  data sharing.
\newblock {\em Circulation: Cardiovascular Quality and Outcomes},
  12(7):e005122, 2019.

\bibitem{cai2021data}
K.~Cai, X.~Lei, J.~Wei, and X.~Xiao.
\newblock Data synthesis via differentially private markov random fields.
\newblock {\em Proceedings of the VLDB Endowment}, 14(11):2190--2202, 2021.

\bibitem{chen2020gs}
D.~Chen, T.~Orekondy, and M.~Fritz.
\newblock Gs-wgan: A gradient-sanitized approach for learning differentially
  private generators.
\newblock {\em Advances in Neural Information Processing Systems},
  33:12673--12684, 2020.

\bibitem{chen2017fast}
H.~Chen, K.~Laine, and P.~Rindal.
\newblock Fast private set intersection from homomorphic encryption.
\newblock In {\em Proceedings of the 2017 ACM SIGSAC Conference on Computer and
  Communications Security}, pages 1243--1255, 2017.

\bibitem{chen2020vafl}
T.~Chen, X.~Jin, Y.~Sun, and W.~Yin.
\newblock Vafl: a method of vertical asynchronous federated learning.
\newblock {\em arXiv preprint arXiv:2007.06081}, 2020.

\bibitem{chen2021secureboost+}
W.~Chen, G.~Ma, T.~Fan, Y.~Kang, Q.~Xu, and Q.~Yang.
\newblock Secureboost+: A high performance gradient boosting tree framework for
  large scale vertical federated learning.
\newblock {\em arXiv preprint arXiv:2110.10927}, 2021.

\bibitem{dickens2022order}
C.~Dickens, J.~Thaler, and D.~Ting.
\newblock Order-invariant cardinality estimators are differentially private.
\newblock {\em Advances in Neural Information Processing Systems},
  35:15204--15216, 2022.

\bibitem{dong2013private}
C.~Dong, L.~Chen, and Z.~Wen.
\newblock When private set intersection meets big data: an efficient and
  scalable protocol.
\newblock In {\em Proceedings of the 2013 ACM SIGSAC conference on Computer \&
  communications security}, pages 789--800, 2013.

\bibitem{dwork2006calibrating}
C.~Dwork, F.~McSherry, K.~Nissim, and A.~Smith.
\newblock Calibrating noise to sensitivity in private data analysis.
\newblock In {\em Theory of Cryptography: Third Theory of Cryptography
  Conference, TCC 2006, New York, NY, USA, March 4-7, 2006. Proceedings 3},
  pages 265--284. Springer, 2006.

\bibitem{frigerio2019differentially}
L.~Frigerio, A.~S. de~Oliveira, L.~Gomez, and P.~Duverger.
\newblock Differentially private generative adversarial networks for time
  series, continuous, and discrete open data.
\newblock In {\em ICT Systems Security and Privacy Protection: 34th IFIP TC 11
  International Conference, SEC 2019, Lisbon, Portugal, June 25-27, 2019,
  Proceedings 34}, pages 151--164. Springer, 2019.

\bibitem{gaboardi2014dual}
M.~Gaboardi, E.~J.~G. Arias, J.~Hsu, A.~Roth, and Z.~S. Wu.
\newblock Dual query: Practical private query release for high dimensional
  data.
\newblock In {\em International Conference on Machine Learning}, pages
  1170--1178. PMLR, 2014.

\bibitem{ge2020kamino}
C.~Ge, S.~Mohapatra, X.~He, and I.~F. Ilyas.
\newblock Kamino: Constraint-aware differentially private data synthesis.
\newblock {\em arXiv preprint arXiv:2012.15713}, 2020.

\bibitem{groce2019cheaper}
A.~Groce, P.~Rindal, and M.~Rosulek.
\newblock Cheaper private set intersection via differentially private leakage.
\newblock {\em Proceedings on Privacy Enhancing Technologies}, 2019(3), 2019.

\bibitem{hardt2012simple}
M.~Hardt, K.~Ligett, and F.~McSherry.
\newblock A simple and practical algorithm for differentially private data
  release.
\newblock {\em Advances in neural information processing systems}, 25, 2012.

\bibitem{hardy2017private}
S.~Hardy, W.~Henecka, H.~Ivey-Law, R.~Nock, G.~Patrini, G.~Smith, and
  B.~Thorne.
\newblock Private federated learning on vertically partitioned data via entity
  resolution and additively homomorphic encryption.
\newblock {\em arXiv preprint arXiv:1711.10677}, 2017.

\bibitem{he2022differentially}
Y.~He, X.~Tan, J.~Ni, L.~T. Yang, and X.~Deng.
\newblock Differentially private set intersection for asymmetrical id
  alignment.
\newblock {\em IEEE Transactions on Information Forensics and Security},
  17:3479--3494, 2022.

\bibitem{hu2019fdml}
Y.~Hu, D.~Niu, J.~Yang, and S.~Zhou.
\newblock {FDML}: A collaborative machine learning framework for distributed
  features.
\newblock In {\em Proceedings of the 25th ACM SIGKDD International Conference
  on Knowledge Discovery \& Data Mining}, KDD '19, page 2232–2240, New York,
  NY, USA, 2019. Association for Computing Machinery.

\bibitem{hu2023sok}
Y.~Hu, F.~Wu, Q.~Li, Y.~Long, G.~Garrido, C.~Ge, B.~Ding, D.~Forsyth, B.~Li,
  and D.~Song.
\newblock Sok: Privacy-preserving data synthesis.
\newblock In {\em Proc. IEEE S\&P}, pages 2--2, 2023.

\bibitem{jiang2006secure}
W.~Jiang and C.~Clifton.
\newblock A secure distributed framework for achieving k-anonymity.
\newblock {\em The VLDB journal}, 15:316--333, 2006.

\bibitem{jiang2023distributed}
X.~Jiang, Y.~Zhang, X.~Zhou, and J.~Grossklags.
\newblock Distributed gan-based privacy-preserving publication of
  vertically-partitioned data.
\newblock {\em Proceedings on Privacy Enhancing Technologies}, 2:236--250,
  2023.

\bibitem{jordon2018pate}
J.~Jordon, J.~Yoon, and M.~Van Der~Schaar.
\newblock Pate-gan: Generating synthetic data with differential privacy
  guarantees.
\newblock In {\em International conference on learning representations}, 2018.

\bibitem{kairouz2021advances}
P.~Kairouz, H.~B. McMahan, B.~Avent, A.~Bellet, M.~Bennis, A.~N. Bhagoji,
  K.~Bonawitz, Z.~Charles, G.~Cormode, R.~Cummings, et~al.
\newblock Advances and open problems in federated learning.
\newblock {\em Foundations and Trends{\textregistered} in Machine Learning},
  14(1--2):1--210, 2021.

\bibitem{kifer2009attacks}
D.~Kifer.
\newblock Attacks on privacy and definetti's theorem.
\newblock In {\em Proceedings of the 2009 ACM SIGMOD International Conference
  on Management of data}, pages 127--138, 2009.

\bibitem{kolesnikov2016efficient}
V.~Kolesnikov, R.~Kumaresan, M.~Rosulek, and N.~Trieu.
\newblock Efficient batched oblivious prf with applications to private set
  intersection.
\newblock In {\em Proceedings of the 2016 ACM SIGSAC Conference on Computer and
  Communications Security}, pages 818--829, 2016.

\bibitem{krawczyk2005hmqv}
H.~Krawczyk.
\newblock Hmqv: A high-performance secure diffie-hellman protocol.
\newblock In {\em Annual international cryptology conference}, pages 546--566.
  Springer, 2005.

\bibitem{li2014differentially}
H.~Li, L.~Xiong, and X.~Jiang.
\newblock Differentially private synthesization of multi-dimensional data using
  copula functions.
\newblock In {\em Advances in database technology: proceedings. International
  conference on extending database technology}, volume 2014, page 475. NIH
  Public Access, 2014.

\bibitem{li2014dpsynthesizer}
H.~Li, L.~Xiong, L.~Zhang, and X.~Jiang.
\newblock Dpsynthesizer: differentially private data synthesizer for privacy
  preserving data sharing.
\newblock In {\em Proceedings of the VLDB Endowment International Conference on
  Very Large Data Bases}, volume~7, page 1677. NIH Public Access, 2014.

\bibitem{li2017differential}
N.~Li, M.~Lyu, D.~Su, and W.~Yang.
\newblock {\em Differential privacy: From theory to practice}.
\newblock Springer, 2017.

\bibitem{li2021dpsyn}
N.~Li, Z.~Zhang, and T.~Wang.
\newblock Dpsyn: Experiences in the nist differential privacy data synthesis
  challenges.
\newblock {\em arXiv preprint arXiv:2106.12949}, 2021.

\bibitem{li2024performance}
Z.~Li, B.~Ding, L.~Yao, Y.~Li, X.~Xiao, and J.~Zhou.
\newblock Performance-based pricing for federated learning via auction.
\newblock {\em Proc. VLDB Endowment}, 17(6):1269--1282, 2024.

\bibitem{li2022vldb}
Z.~Li, B.~Ding, C.~Zhang, N.~Li, and J.~Zhou.
\newblock Federated matrix factorization with privacy guarantee.
\newblock {\em Proc. VLDB Endow.}, 15(4):900–913, dec 2021.

\bibitem{li2023differentially}
Z.~Li, T.~Wang, and N.~Li.
\newblock Differentially private vertical federated clustering.
\newblock {\em Proceedings of the VLDB Endowment}, 16(6):1277--1290, 2023.

\bibitem{liu2019communication}
Y.~Liu, Y.~Kang, X.~Zhang, L.~Li, Y.~Cheng, T.~Chen, M.~Hong, and Q.~Yang.
\newblock A communication efficient collaborative learning framework for
  distributed features.
\newblock {\em arXiv preprint arXiv:1912.11187}, 2019.

\bibitem{liu2020federated-forest}
Y.~Liu, Y.~Liu, Z.~Liu, Y.~Liang, C.~Meng, J.~Zhang, and Y.~Zheng.
\newblock Federated forest.
\newblock {\em IEEE Transactions on Big Data}, (01):1--1, 2020.

\bibitem{manton2010national}
K.~G. Manton.
\newblock National long-term care survey: 1982, 1984, 1989, 1994, 1999, and
  2004.
\newblock {\em Inter-university Consortium for Political and Social Research},
  2010.

\bibitem{mckenna2021winning}
R.~McKenna, G.~Miklau, and D.~Sheldon.
\newblock Winning the nist contest: A scalable and general approach to
  differentially private synthetic data.
\newblock {\em arXiv preprint arXiv:2108.04978}, 2021.

\bibitem{mckenna2019graphical}
R.~McKenna, D.~Sheldon, and G.~Miklau.
\newblock Graphical-model based estimation and inference for differential
  privacy.
\newblock In {\em International Conference on Machine Learning}, pages
  4435--4444. PMLR, 2019.

\bibitem{mcmahan2017dpfedavg}
B.~McMahan, E.~Moore, D.~Ramage, S.~Hampson, and B.~A.~y. Arcas.
\newblock {Communication-Efficient Learning of Deep Networks from Decentralized
  Data}.
\newblock In A.~Singh and J.~Zhu, editors, {\em Proceedings of the 20th
  International Conference on Artificial Intelligence and Statistics},
  volume~54 of {\em Proceedings of Machine Learning Research}, pages
  1273--1282, USA, 20--22 Apr 2017. PMLR.

\bibitem{mcmahan2018dp-rnn}
H.~B. McMahan, D.~Ramage, K.~Talwar, and L.~Zhang.
\newblock Learning differentially private recurrent language models.
\newblock In {\em International Conference on Learning Representations}.
  OpenReview.net, 2018.

\bibitem{mironov2017renyi}
I.~Mironov.
\newblock R{\'e}nyi differential privacy.
\newblock In {\em 2017 IEEE 30th computer security foundations symposium
  (CSF)}, pages 263--275. IEEE, 2017.

\bibitem{mohammed2013secure}
N.~Mohammed, D.~Alhadidi, B.~C. Fung, and M.~Debbabi.
\newblock Secure two-party differentially private data release for vertically
  partitioned data.
\newblock {\em IEEE transactions on dependable and secure computing},
  11(1):59--71, 2013.

\bibitem{mohammed2011anonymity}
N.~Mohammed, B.~C. Fung, and M.~Debbabi.
\newblock Anonymity meets game theory: secure data integration with malicious
  participants.
\newblock {\em The VLDB Journal}, 20:567--588, 2011.

\bibitem{pardau2018california}
S.~L. Pardau.
\newblock The california consumer privacy act: Towards a european-style privacy
  regime in the united states.
\newblock {\em J. Tech. L. \& Pol'y}, 23:68, 2018.

\bibitem{ren2022ldp}
X.~Ren, L.~Shi, W.~Yu, S.~Yang, C.~Zhao, and Z.~Xu.
\newblock Ldp-ids: Local differential privacy for infinite data streams.
\newblock In {\em Proc. SIGMOD}, pages 1064--1077, 2022.

\bibitem{ren2024belt}
X.~Ren, S.~Yang, C.~Zhao, J.~McCann, and Z.~Xu.
\newblock Belt and brace: When federated learning meets differential privacy.
\newblock {\em arXiv preprint arXiv:2404.18814}, 2024.

\bibitem{ren2018lopub}
X.~Ren, C.-M. Yu, W.~Yu, S.~Yang, X.~Yang, J.~A. McCann, and S.~Y. Philip.
\newblock Lopub: high-dimensional crowdsourced data publication with local
  differential privacy.
\newblock {\em IEEE Transactions on Information Forensics and Security},
  13(9):2151--2166, 2018.

\bibitem{ridgeway2021challenge}
D.~Ridgeway, M.~F. Theofanos, T.~W. Manley, and C.~Task.
\newblock Challenge design and lessons learned from the 2018 differential
  privacy challenges.
\newblock Technical report, Technical Report NIST Technical Note 2151, National
  Institute of Standards~…, 2021.

\bibitem{steven2015ipums}
S.~Ruggles, K.~Genadek, G.~Ronald, G.~Josiah, and M.~Sobek.
\newblock Ipums usa: Version 6.0.
\newblock Technical report, Minneapolis: University of Minnesota, 2015.

\bibitem{smith2020flajolet}
A.~Smith, S.~Song, and A.~Guha~Thakurta.
\newblock The flajolet-martin sketch itself preserves differential privacy:
  Private counting with minimal space.
\newblock {\em Advances in Neural Information Processing Systems},
  33:19561--19572, 2020.

\bibitem{su2016differentially}
S.~Su, P.~Tang, X.~Cheng, R.~Chen, and Z.~Wu.
\newblock Differentially private multi-party high-dimensional data publishing.
\newblock In {\em Proc. IEEE ICDE}, pages 205--216, 2016.

\bibitem{sweeney2002k}
L.~Sweeney.
\newblock k-anonymity: A model for protecting privacy.
\newblock {\em International journal of uncertainty, fuzziness and
  knowledge-based systems}, 10(05):557--570, 2002.

\bibitem{tang2019differentially}
P.~Tang, X.~Cheng, S.~Su, R.~Chen, and H.~Shao.
\newblock Differentially private publication of vertically partitioned data.
\newblock {\em IEEE transactions on dependable and secure computing},
  18(2):780--795, 2019.

\bibitem{vietri2020new}
G.~Vietri, G.~Tian, M.~Bun, T.~Steinke, and S.~Wu.
\newblock New oracle-efficient algorithms for private synthetic data release.
\newblock In {\em International Conference on Machine Learning}, pages
  9765--9774. PMLR, 2020.

\bibitem{voigt2017eu}
P.~Voigt and A.~Von~dem Bussche.
\newblock The eu general data protection regulation (gdpr).
\newblock {\em A Practical Guide, 1st Ed., Cham: Springer International
  Publishing}, 10(3152676):10--5555, 2017.

\bibitem{wang2019answering}
T.~Wang, B.~Ding, J.~Zhou, C.~Hong, Z.~Huang, N.~Li, and S.~Jha.
\newblock Answering multi-dimensional analytical queries under local
  differential privacy.
\newblock In {\em Proceedings of the 2019 International Conference on
  Management of Data}, pages 159--176, 2019.

\bibitem{wang2019locally}
T.~Wang, X.~Yang, X.~Ren, W.~Yu, and S.~Yang.
\newblock Locally private high-dimensional crowdsourced data release based on
  copula functions.
\newblock {\em IEEE Transactions on Services Computing}, 15(2):778--792, 2019.

\bibitem{webank}
WeBank.
\newblock Webank use case.
\newblock
  \url{https://www.fedai.org/cases/a-case-of-traffic-violations-insurance-using-federated-learning/},
  2022.

\bibitem{wong2007minimality}
R.~C.-W. Wong, A.~W.-C. Fu, K.~Wang, and J.~Pei.
\newblock Minimality attack in privacy preserving data publishing.
\newblock In {\em Proceedings of the 33rd international conference on Very
  large data bases}, pages 543--554, 2007.

\bibitem{wu2020privacy}
Y.~Wu, S.~Cai, X.~Xiao, G.~Chen, and B.~C. Ooi.
\newblock Privacy preserving vertical federated learning for tree-based models.
\newblock {\em arXiv preprint arXiv:2008.06170}, 2020.

\bibitem{wu2023vertibench}
Z.~Wu, J.~Hou, and B.~He.
\newblock Vertibench: Advancing feature distribution diversity in vertical
  federated learning benchmarks.
\newblock {\em arXiv preprint arXiv:2307.02040}, 2023.

\bibitem{xie2022improving}
C.~Xie, P.-Y. Chen, C.~Zhang, and B.~Li.
\newblock Improving privacy-preserving vertical federated learning by efficient
  communication with admm.
\newblock {\em arXiv preprint arXiv:2207.10226}, 2022.

\bibitem{xie2023federatedscope}
Y.~Xie, Z.~Wang, D.~Gao, D.~Chen, L.~Yao, W.~Kuang, Y.~Li, B.~Ding, and
  J.~Zhou.
\newblock Federatedscope: A flexible federated learning platform for
  heterogeneity.
\newblock {\em Proc. VLDB Endow.}, 16(5):1059--1072, 2023.

\bibitem{zhang2017privbayes}
J.~Zhang, G.~Cormode, C.~M. Procopiuc, D.~Srivastava, and X.~Xiao.
\newblock Privbayes: Private data release via bayesian networks.
\newblock {\em ACM Transactions on Database Systems (TODS)}, 42(4):1--41, 2017.

\bibitem{zhang2018differentially}
X.~Zhang, S.~Ji, and T.~Wang.
\newblock Differentially private releasing via deep generative model (technical
  report).
\newblock {\em arXiv preprint arXiv:1801.01594}, 2018.

\bibitem{zhang2020additively}
Y.~Zhang and H.~Zhu.
\newblock Additively homomorphical encryption based deep neural network for
  asymmetrically collaborative machine learning.
\newblock {\em arXiv preprint arXiv:2007.06849}, 2020.

\bibitem{zhang2018calm}
Z.~Zhang, T.~Wang, N.~Li, S.~He, and J.~Chen.
\newblock Calm: Consistent adaptive local marginal for marginal release under
  local differential privacy.
\newblock In {\em Proc. ACM CCS}, pages 212--229, 2018.

\bibitem{zhang2021privsyn}
Z.~Zhang, T.~Wang, N.~Li, J.~Honorio, M.~Backes, S.~He, J.~Chen, and Y.~Zhang.
\newblock $\{$PrivSyn$\}$: Differentially private data synthesis.
\newblock In {\em 30th USENIX Security Symposium (USENIX Security 21)}, pages
  929--946, 2021.

\bibitem{zhao2023fullversion}
F.~Zhao, Z.~Li, X.~Ren, B.~Ding, S.~Yang, and Y.~Li.
\newblock Dp-vertical-data-synthesis.
\newblock
  \url{https://anonymous.4open.science/r/DP-vertical-data-synthesis-8871/},
  2023.

\bibitem{zhao2020latent}
F.~Zhao, X.~Ren, S.~Yang, Q.~Han, P.~Zhao, and X.~Yang.
\newblock Latent dirichlet allocation model training with differential privacy.
\newblock {\em IEEE Transactions on Information Forensics and Security},
  16:1290--1305, 2020.

\end{thebibliography}
}
\appendix
\section{Appendix}
\subsection{Additional experimental results}
\label{sec:add_exp}

\noindent\underline{\textbf{Effect of \texttt{LocMRF}.}}
Data parties generate local \texttt{MRF}s to help infer the intra-party marginals. However, simply using \texttt{LocEnc} and \texttt{CarEst} can also achieve the marginal estimation. In Figure~\ref{fig:locMRF}, we compare TVD on NLTCS when synthesizing data with and without using \texttt{LocMRF} respectively. Specifically, different \texttt{LocEnc} methods are considered in \texttt{LocMRF}, which are labeled as \texttt{"FM}+\texttt{LocMRF"} and \texttt{"FO}+\texttt{LocMRF"}, respectively. %\texttt{LocMRF} is employed, with \texttt{FM} and \texttt{FO} referring to the respective \texttt{LocEnc} methods used.
As shown, for both \texttt{"FM}+\texttt{LocMRF"} and \texttt{"FO}+\texttt{LocMRF"}, the TVD results are smaller than those when simply using the sketch or FO-based \texttt{LocEnc}. This demonstrates the effectiveness of \texttt{LocMRF}. Furthermore, we also find that using \texttt{LocMRF} can dramatically reduce the variances of the generated TVD results. This is reasonable since \texttt{LocMRF} can capture the correlations among local attributes, thereby reducing the uncertainty when 
estimating the intra-party marginals. 

\noindent\underline{\textbf{Effect of histogram recovery.}} 
% In Section~\ref{subsec:dim_rec}, we propose to conduct attribute binning and histogram recovery (\texttt{HisRec}) to improve cross-party marginal estimation. During binning, a part of privacy budget is used to calculate the value distribution within each bin. One might question the effectiveness of \texttt{HisRec}. To address this, 
We demonstrate the effectiveness of our proposed histogram recovery (\texttt{HisRec}) by comparing it to the baseline that generates high-dimensional histogram via uniformly allocating the count in each cell of estimated low-dimensional histogram (denoted as \texttt{UniSam}) to the corresponding multiple cells of the high-dimensional histogram. 
For a fair comparison, the privacy budget for sanitizing the value distributions in \texttt{HisRec} is allocated to \texttt{LocEnc} in \texttt{UniSam}.
As shown in Figure~\ref{fig:adult_uni}, \texttt{FM+HisRec} yields superior TVD results compared to \texttt{FM+UniSam}, which demonstrates the effectiveness of \texttt{HisRec} when being used in conjunction with FM sketch-based \texttt{LocEnc} and \texttt{CarEst} approaches. However, we also find that \texttt{FO+HisRec} performs closely to \texttt{FO+UniSam} when the privacy budget $\epsilon<3.2$ and even worse when $\epsilon=3.2$, that's mainly because the low-dimensional histogram estimated by FO-based \texttt{CarEst} is too noisy and the noise has dominated the estimation error. Then without prior knowledge of true distribution, uniform allocating shows an advantage as an impartial method.

% Moreover, the advantages of \texttt{HisRec} become more pronounced as the privacy cost increases. This is because the errors introduced by \texttt{LocEnc} decrease with higher privacy costs, making the histogram transformation procedure the primary source of noise. \texttt{HisRec} ensures consistency between the transformed high-dimensional histogram and the noisy value distributions of the raw data. As the noise decreases, the consistency with the actual value distribution improves, leading to more significant advantages of \texttt{HisRec} over \texttt{UniSam}.

\noindent\underline{\textbf{Effect of consistency enforcement.}} As discussed in Section~\ref{subsec:consis}, we propose to improve marginal estimation by ensuring consistency among the histograms estimated by \texttt{CarEst} and local \texttt{MRF}s. In Figure~\ref{fig:adult_cons}, we compare TVD on the synthetic Adult dataset obtained with (denoted as \texttt{+Consis}) and without enforcing consistency. As shown, the TVD values under \texttt{+Consis} settings are consistently smaller than those without ensuring consistency, demonstrating the effectiveness of consistency enforcement. Furthermore, the TVD results become closer within a lower privacy regime (larger budget). This is due to the improved accuracy of both \texttt{CarEst} and local \texttt{MRF}s with the increased privacy budget, reducing the advantage of the consistency enforcement procedure. Therefore, enforcing consistency becomes more necessary in a higher privacy regime.

\begin{figure*}[t]
	\centering
	\includegraphics[width=0.98\textwidth]{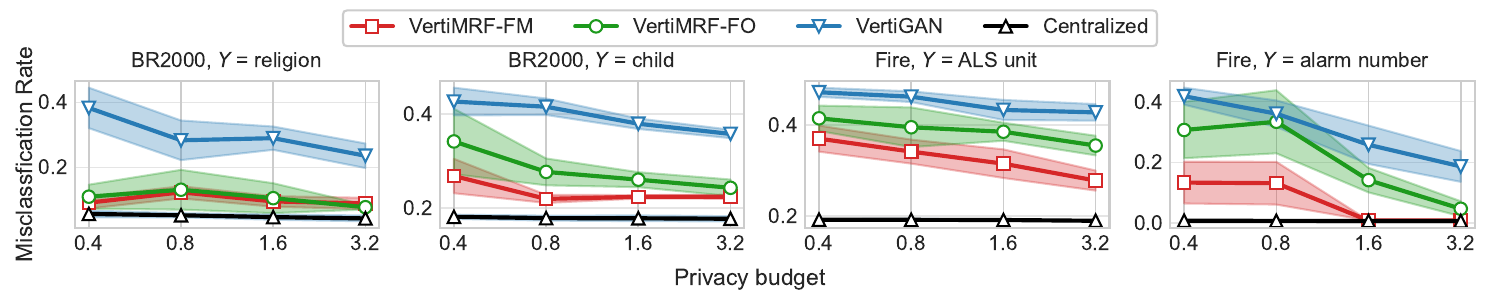}
    \caption{SVM misclassification rate vs. privacy budget $\epsilon$ on BR2000 and Fire datasets. \label{fig:br2000_svm_accuracy}.}
    % \vspace{-0.1in}
\end{figure*}

\begin{figure}[t]
	\centering
	\includegraphics[width=0.48\textwidth]{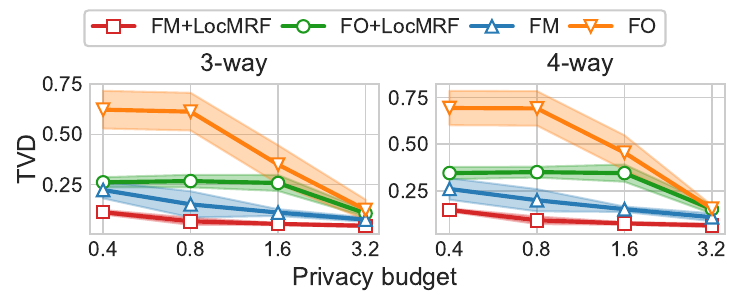}
    % \vspace{-0.25in}
    \caption{Effect of \texttt{LocMRF} on NLTCS. \label{fig:locMRF}}
    % \vspace{-0.15in}
\end{figure}

\begin{figure}[t]
	\centering
	\includegraphics[width=0.48\textwidth]{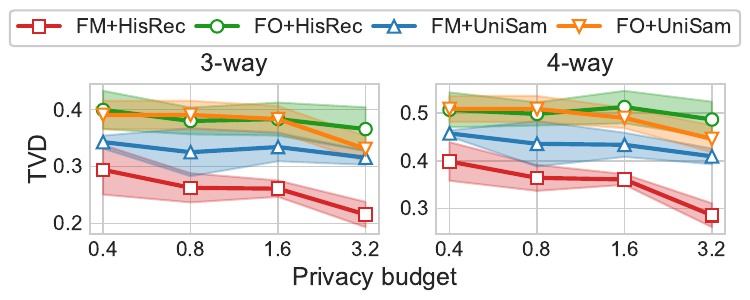}
    % \vspace{-0.25in}
    \caption{Effect of histogram recovery (\texttt{HisRec}) on Adult. \label{fig:adult_uni}}
    % \vspace{-0.15in}
\end{figure}

\begin{figure}[t]
	\centering
	\includegraphics[width=0.48\textwidth]{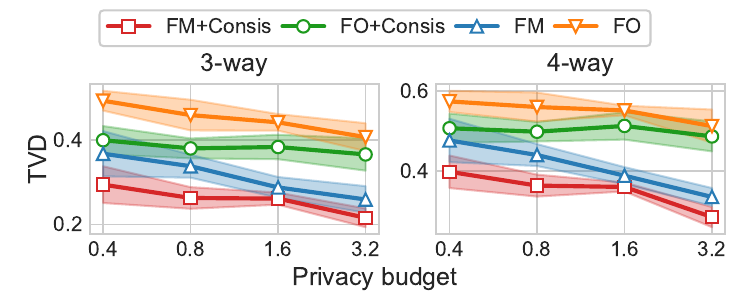}
    \caption{Effect of enforcing consistency on Adult. \label{fig:adult_cons}}
\end{figure}

\subsection{proof of Theorem~\ref{the:rr-error}}\label{proof:the3}
\begin{proof}
First of all, we should notice that GRR achieves unbounded DP~\cite{li2017differential} which considers the neighboring dataset by replacing one record. It has been shown that any algorithm satisfying $\epsilon$ unbounded DP also satisfies $2\epsilon$ bounded DP, where bounded DP considers neighboring datasets obtained by adding or removing a single record. In this paper, we consider bounded DP for consistency. Thus, given $\epsilon'$, each perturbation in Equation~(\ref{equ:grr}) should satisfy $\frac{\epsilon'}{2}$-DP. The privacy guarantee of the FO-based \texttt{LocEnc} procedure can be obtained by applying the sequential composition of DP, resulting in an overall privacy guarantee of $d\epsilon'/2$, where $d$ represents the number of attributes in each record. However, Lemma~\ref{lem:com_rdp} demonstrates that RDP provides an alternative bound for the composition of multiple DP algorithms, that is $(4\frac{\epsilon'}{2}\sqrt{2d\log(1/\delta)},\delta)$-DP, where $0<\delta<1$ and $\log(1/\delta)\geq n(\frac{\epsilon'}{2})^2$. To obtain the tighter bound, we take the minimum between the two bounds, as stated in the theorem. The variance bound can be directly obtained from proposition 10 of~\cite{wang2019answering}.
\end{proof}

\subsection{Proof of Theorem~\ref{the:fm_priv}}
\label{proof:the4}

% \zitao{consider only one hash key first. Then adding/removing one user can at most affect $d_i$ (number of attributes) cells in $\mathcal{M}_i$. Then compute the probability that the user's hashed id is not the max of the given cells, which should be upper bounded by considering only the phantom elements and lower bounds.}
\begin{proof}
Let $D$ and $D'$ be neighboring datasets satisfying  $D\nabla D' = X_{id} = \left\{v_{id}^{{1}},..., v_{id}^{{d}} \right\}$, where $id$ denotes the record- index of $X_{id}$,  $v_{id}^{{j}}$ is the corresponding attribute value of $A^j$. Let $f$ be the sketch-based \texttt{LocEnc} algorithm which maps $t$ hash keys and input dataset to $t$ set of sketch tuples $$\left\{\mathcal{M}^{(h)}\triangleq\left\{\mathcal{M}^{(h)}_j\triangleq\left(\alpha^{(h)}_{v_1^j}, ..., \alpha^{(h)}_{v_{u_j}^j}\right)\vert j\in [d]\right\}\vert h\in [t]\right\}.$$ where $\alpha^{(h)}_{v_i^j}$ denotes the sketch for $A^j$ taking value $v_i^j$ generated by the hash key $\xi_{h}$.

We first calculate the privacy cost when applying a hash key $\xi_h$ to the overall input dataset and returning sketch tuples $\mathcal{M}^{(h)}$.  $\mathcal{M}^{(h)}$ has $d$ sketch tuples and $\sum_{i=1}^{d}u_i$ sketches in total. Since $X_{id}$ can only take one value $v_{id}^j$ of each attribute $A^j$, then there should also be one sketch $\alpha_{v_{id}^j}^{(h)}$ in $\mathcal{M}_{j}^{(h)}$ may be different for $D$ and $D'$. Therefore, according to the definition of RDP, it holds that 

% \begin{equation}
% \begin{align}\label{equ:total_pri}
% &\quad\exp\left[(\lambda -1)D_{\lambda}\left(f(D, \xi_h)\vert f(D', \xi_h)\right)\right]\\
% &= \sum_{\mathbb{R}^{\sum_{i=1}^{d}u_i}}{Pr[\mathcal{M}^{(h)}]}^{\lambda}{Pr^{'}[\mathcal{M}^{(h)}]}^{1-\lambda}\mathrm{d}\mathcal{M}^{(h)}\\
% &=\int_{\mathbb{R}^{d}}{\left[Pr[\alpha_{A^1, v_{id}^1}^{(h)}]\prod_{1<i\leq d} {Pr[\alpha_{a_i, v_{id}^i}^{(h)}|\{\alpha_{A^t, v_{id}^t}^{(h)},t<i\}]}\right]}^{\lambda}\cdot\\
% &\quad \left[Pr^{\prime}[\alpha_{A^1, v_{id}^1}^{(h)}]\prod_{1<i\leq d}{Pr^{\prime}}[\alpha_{a_i, v_{id}^i}^{(h)}|\{\alpha_{A^t, v_{id}^t}^{(h)},t<i\}]\right]^{1-\lambda}\prod_{i=1}^{d}{\mathrm{d}\alpha_{a_i, v_{id}^i}^{(h)}}\cdot\\
% &\quad \underbrace{\int_{\mathbb{R}^{\sum_{i=1}^{d}(u_i-1)}}\left[Pr[\mathcal{M}_{\neg}^{(h)}|\vec{\alpha}]\right]^{\lambda}\left[Pr^{\prime}[\mathcal{M}_{\neg}^{(h)}|\vec{\alpha}]\right]^{1-\lambda}{\mathrm{d}\mathcal{M}_{\neg}^{(h)}}}_{=1}
% % &= \int_{R^{\prod|\mathcal{X}_i|}}{\prod_i {Pr[\mathcal{M}_i|D,s]}}^{\lambda}\prod_i{Pr^{'}[\mathcal{M}_i|D',s]}^{1-\lambda}\prod_{i}{\mathrm{d}\mathcal{M}_i}\\
% % &= \int_{R^{\prod|\mathcal{X}_{i}|}}\prod_{i} \left\{{Pr[\mathcal{M}_i|D,s]}^{\lambda}{Pr^{'}[\mathcal{M}_i|D',s]}^{1-\lambda}\mathrm{d}\mathcal{M}_i\right\}\\
% % &\leq \prod_{i}\exp\left[(\alpha-1)\epsilon_i\right] = \exp\left[(\alpha-1)\sum_i{\epsilon_i}\right],
% \end{align}
% \end{equation}
\begin{align}\label{equ:total_pri}
&\quad\exp\left[(\lambda -1)D_{\lambda}\left(f(D, \xi_h)\vert f(D', \xi_h)\right)\right]\\
&= \sum_{\mathcal{M}^{(h)}}{Pr[\mathcal{M}^{(h)}]}^{\lambda}{Pr^{'}[\mathcal{M}^{(h)}]}^{1-\lambda}\\
&=\sum_{\alpha_{v_{id}^1}^{(h)}=0}^{\infty}...\sum_{\alpha_{v_{id}^d}^{(h)}=0}^{\infty}\{{[Pr[\alpha_{v_{id}^1}^{(h)}]\prod_{1<i\leq d} {Pr[\alpha_{v_{id}^i}^{(h)}|\{\alpha_{v_{id}^t}^{(h)},t<i\}]}]}^{\lambda}\cdot\\
&\quad [Pr^{\prime}[\alpha_{v_{id}^1}^{(h)}]\prod_{1<i\leq d}{Pr^{\prime}}[\alpha_{v_{id}^i}^{(h)}|\{\alpha_{v_{id}^t}^{(h)},t<i\}]]^{1-\lambda}\}\cdot\\
&\quad\underbrace{\sum_{\mathcal{M}_{\neg}^{(h)}}[Pr[\mathcal{M}_{\neg}^{(h)}|\vec{\alpha}]]^{\lambda}[Pr^{\prime}[\mathcal{M}_{\neg}^{(h)}|\vec{\alpha}]]^{1-\lambda}}_{=1}
\end{align}
where the second equality follows the joint distribution formula, $\vec{\alpha} \triangleq \{\alpha_{v_{id}^i}^{(h)}, 1\leq i\leq d\}$ and $\mathcal{M}_{\neg}^{(h)}$ denotes other sketches in $\mathcal{M}^{(h)}$ besides $\vec{\alpha}$. 

% Now we bound $$\int_{\mathbb{R}^{d}}Pr[\alpha_{a_i, v_{id}^i}^{(h)}|\{\alpha_{A^t, v_{id}^t}^{(h)},t<i\}]^{\lambda}\cdot Pr^{\prime}[\alpha_{a_i, v_{id}^i}^{(h)}|\{\alpha_{A^t, v_{id}^t}^{(h)},t<i\}]^{1-\lambda}\mathrm{d}\alpha_{a_i, v_{id}^i}^{(h)}.$$

% Before further analysis, we first provide an observation that each element in $\vec{\alpha}$ takes the same value for $D$ and $D'$ if and only if $\mathcal{H}_{\xi_{h}}(id) \neq \alpha_{A^t, v_{id}^t}^{(h)}, \forall t$. That means $\mathcal{H}_{\xi_{h}}(id)$ is not the maximum among the set of  $\{\mathcal{H}_{\xi_{h}}(id), Y_{1}, ..., Y_{k_p},b\}$ where $Y_{1}, ..., Y_{k_p}$ are independently sampled Geometric random variables,
% $k_p = \lceil \frac{1}{e^{\epsilon^{\prime}-1}} \rceil$ and $b = \lceil\log_{1+\gamma}{\frac{1}{1-e^{-\epsilon^{\prime}}}}\rceil$. Such observation holds since $X_{id}$ exists in only one of $D$ and $D'$.

Now consider term $Pr\left[\alpha_{v_{id}^i}^{(h)}|\left\{\alpha_{v_{id}^t}^{(h)},t<i\right\}\right]$, there are two cases: 
\begin{itemize}[leftmargin = *]
\item  $\forall t, \alpha_{v_{id}^i}^{(h)} \neq \alpha_{v_{id}^t}^{(h)}$. In such case, since $\mathcal{H}_{\xi}$ map distinct elements to independent variables and the $k_p$ phantom elements are independently sampled, then $\alpha_{v_{id}^i}^{(h)}$ is independent of $\alpha_{v_{id}^t}^{(h)}, \forall t$. That indicates $$Pr\left[\alpha_{v_{id}^i}^{(h)}|\left\{\alpha_{v_{id}^t}^{(h)},t<i\right\}\right] = Pr\left[\alpha_{v_{id}^i}^{(h)}\right].$$
\item $\exists t, s.t., \alpha_{v_{id}^i}^{(h)} = \alpha_{v_{id}^t}^{(h)}$. In such case, it should hold that $\alpha_{v_{id}^i}^{(h)}\geq\mathcal{H}_{\xi_{h}}(id)$. That indicates 
\begin{align*}
Pr\left[\alpha_{v_{id}^i}^{(h)}|\left\{\alpha_{v_{id}^t}^{(h)},t<i\right\}\right]
= Pr^{\prime}\left[\alpha_{v_{id}^i}^{(h)}|\left\{\alpha_{v_{id}^t}^{(h)},t<i\right\}\right].
\end{align*}
The left side of above Equation is the probability that $\alpha_{v_{id}^i}^{(h)}$ is the maximal among all elements in the set of hashed record ids and sampled geometric random variables on $D$. Since we have known that $\mathcal{H}_{\xi_{h}}(id)$ is not or not the only one maximal element in the set, then we can just consider other hashed ids and sampled variables. The ids are same for $D$ and $D^{\prime}$ and each of the variables are i.i.d sampled from the same distribution, which can easily derive the equality of above equation.
\end{itemize}

W.l.o.g., we assume there are $s$ terms $\{Pr\left[\alpha_{v_{id}^j}^{(h)}|\left\{\alpha_{v_{id}^t}^{(h)},t<j\right\}\right], (d-s+1)\leq j\leq d\}$ satisfying the second case. Then, Equation~(\ref{equ:total_pri}) can be bounded by:
\begin{align}
&\quad\exp\left[(\lambda -1)D_{\lambda}\left(f(D, \xi_h)\vert f(D', \xi_h)\right)\right]\\
&= \sum_{\alpha_{v_{id}^1}^{(h)}=0}^{\infty}...\sum_{\alpha_{v_{id}^{d-s}}^{(h)}=0}^{\infty} {\left[\prod_{i=1}^{d-s} {Pr[\alpha_{v_{id}^i}^{(h)}]}\right]}^{\lambda}\cdot {\left[\prod_{i=1}^{d-s} {Pr^{\prime}[\alpha_{v_{id}^i}^{(h)}]}\right]}^{1-\lambda}\\
&= \prod_{i=1}^{d-s}\underbrace{\left\{\sum_{\alpha_{v_{id}^i}^{(h)}=0}^{\infty}{\left[{Pr[\alpha_{v_{id}^i}^{(h)}]}\right]}^{\lambda}{\left[{Pr^{\prime}[\alpha_{v_{id}^i}^{(h)}]}\right]}^{1-\lambda}\right\}}_{term (i)}
\end{align}
Lemma~\ref{lemma:dp-fm} demonstrates a statistical bound of $\epsilon^{'}$ under DP framework. According to the definition of RDP and the translation with DP, it holds that $term (i) \leq \exp{[(\lambda -1)(2\lambda(\epsilon^{\prime})^2)]}$. Then we can derive that 
\begin{align}
&\quad\exp\left[(\lambda -1)D_{\lambda}\left(f(D, \xi_h)\vert f(D', \xi_h)\right)\right]\\
&\leq \exp{[(\lambda -1)(2(d-s)\lambda(\epsilon^{\prime})^2)]}\\
&\leq \exp{[(\lambda -1)(2d\lambda(\epsilon^{\prime})^2)]}
\end{align}

So far, we have proved that applying one hash key to map the overall input data satisfies $(\lambda, 2d\lambda(\epsilon^{\prime})^2)$-RDP in a single run. Next, according to the sequential composition theorem of RDP~\cite{mironov2017renyi}, \texttt{LocEnc} algorithm involving $t$ runs of the FM sketch generation process should satisfy $(\lambda, 2td\lambda(\epsilon^{\prime})^2)$-RDP, which can be further translated to $(4\epsilon\sqrt{td\log(1/\delta)}, \delta)$-DP, $\forall \delta<1$ if setting $\alpha\geq 2$.
\end{proof}

\subsection{Proof of Theorem~\ref{the:fm_error}}\label{proof:the5}
Our proof is based on a lemma that bounds the error of the cardinality of a multi-set estimated by DP FM sketching algorithm shown in Algorithm~\ref{alg:dpfm}.

\begin{lemma}~\label{lem:uti-dpfm}
Let $k_{FM}$ be the estimated cardinality by Algorithm~\ref{alg:dpfm} with inputs $\gamma, \epsilon, \delta,\beta \in \left(0, 1\right)$, using $t=\frac{100\sqrt{\log{(1/\beta)}}}{\gamma^{2}}$ repeats, then for each multi-set $\mathcal{X} \subset u$, it holds that 
\begin{align}
\frac{\vert \mathcal{X}\vert}{1+\gamma}-O\leq k_{FM} \leq \left(1+\gamma \right)\cdot\vert\mathcal{X}\vert+C
\end{align}
with probability at least $1-\beta$, where $C = O(\frac{\log^{1/2}(1/\delta)\log^{1/4}(1/\beta)}{\epsilon})$.
\end{lemma}

\begin{proof}
We first bound each cardinality $\hat{T}_M[\mathbf{v}]$ estimated by FM sketch. As shown in the Algorithm~\ref{alg:fm_CarEst}, we compute each $\hat{T}_M[\mathbf{v}]$ using the inclusion-exclusion principle and the megeable property of sketch, that is $|A\cap B| =  \hat{n}- |\bar{A}\cup \bar{B}|$, where $\hat{n}$ denotes the noisy data number sanitized by adding a Laplacian noise $\hat{N}$. Combining with lemma~\ref{lem:uti-dpfm}, we can derive that
% \begin{equation}
\begin{align}
\label{equ:fm-error-com} & \hat{n} - ( n-T_M[\mathbf{v}])\cdot(1+\gamma)-C \leq \hat{T}_M[\mathbf{v}] \leq   \hat{n} - \frac{n-T_M[\mathbf{v}]}{1+\gamma}+C\\
\label{equ:fm-error-com1} &-\gamma n+(1+\gamma)T_M[\mathbf{v}]+\hat{N}-C \leq \hat{T}_M[\mathbf{v}] \leq \frac{\gamma}{1+\gamma} n + \frac{T_M[\mathbf{v}]}{1+\gamma}+\hat{N}+C
\end{align}

By subtracting $T_M[\mathbf{v}]$ for both sides of Equation~\ref{equ:fm-error-com1}, we can obtain that:
\begin{footnotesize}
\begin{align}
-\gamma n+\gamma T_M[\mathbf{v}]+\hat{N}-C \leq \hat{T}_M[\mathbf{v}] - T_M[\mathbf{v}] \leq \frac{\gamma}{1+\gamma} n - \frac{\gamma T_M[\mathbf{v}]}{1+\gamma}+\hat{N}+C
\end{align}
\end{footnotesize}
 By taking the absolute value for both sides and dividing them by $T_M[\mathbf{v}]$, we can derive that:
 \begin{align}
\frac{\vert\hat{T}_M[\mathbf{v}] - T_M[\mathbf{v}]\vert}{T_M[\mathbf{v}]}  &\leq \max \{\frac{\gamma}{1+\gamma} (\frac{n}{T_M[\mathbf{v}]}-1) + \frac{\hat{N}+C}{T_M[\mathbf{v}]}, \\
 &\gamma (\frac{n}{T_M[\mathbf{v}]}-1) -\frac{\hat{N}-C}{T_M[\mathbf{v}]}\}\\
 &\leq \gamma (\frac{n}{T_M[\mathbf{v}]}-1) +\frac{\hat{N}+C}{T_M[\mathbf{v}]}
\end{align}
According to Lemma~\ref{lem:uti-dpfm}, the above bound holds with probability $1-\beta$, and $C = O(\frac{\log^{1/2}(1/\delta)\log^{1/4}(1/\beta)}{\epsilon})$.
\end{proof}

\end{document}